\documentclass{amsart}

\usepackage{amsmath,amssymb,amsthm}
\usepackage{graphicx}
\newtheorem{theorem}{Theorem}[section]
\newtheorem{proposition}[theorem]{Proposition}
\newtheorem{lemma}[theorem]{Lemma}

\theoremstyle{definition}
\newtheorem{definition}[theorem]{Definition}

\newtheorem{corollary}[theorem]{Corollary}
\theoremstyle{remark}
\newtheorem{remark}[theorem]{Remark}

\numberwithin{equation}{section}

\def\d{{\rm d}}
\def\e{{\rm e}}
\def\i{{\rm i}}
\def\Id{{\rm Id}}
\renewcommand{\Re}{\operatorname{Re}}

\def\op{\operatorname{op}}

\def\eps{\varepsilon}
\def\C{{\mathbb C}}
\def\N{{\mathbb N}}
\def\R{{\mathbb R}}
\def\Rd{{{\mathbb R}^d}}
\def\Rdd{{{\mathbb R}^{2d}}}

\def\H{{\mathcal H}}
\def\W{{\mathcal W}}
\def\Sc{{\mathcal S}}

\def\defeq{\mathrel{\mathop:}=}

\def\parentheses#1{{\left(#1\right)}}
\def\brackets#1{{\left[#1\right]}}

\def\norm#1{{\left\|#1\right\|}}
\def\abs#1{{\left|#1\right|}}
\def\ip#1#2{{\left\langle#1,#2\right\rangle}}

\title{A new Phase Space Density for Quantum Expectations}

\thanks{
Support by the German Research Foundation (DFG), Collaborative Research Center SFB-TRR 109, the
graduate program TopMath of the Elite Network of Bavaria and the Simons Visiting Professorship program of the Simons Foundation and the Mathematisches Forschungsinstitut Oberwolfach is gratefully acknowledged.
}
\author{Johannes Keller}
\address{Zentrum Mathematik, Technische Universit\"at M\"unchen, Boltzmannstra\ss e 3, 85748 Garching bei M\"unchen.}
\email{keller@ma.tum.de}
\author{Caroline Lasser}
\address{Zentrum Mathematik, Technische Universit\"at M\"unchen, Boltzmannstra\ss e 3, 85748 Garching bei M\"unchen.}
\email{classer@ma.tum.de}
\author{Tomoki Ohsawa}
\address{Department of Mathematical Sciences, The University of Texas at Dallas, 800 West Campbell Rd, Richardson, TX 75080-3021.}
\email{tomoki@utdallas.edu}


\begin{document}
\pagestyle{myheadings}
\thispagestyle{plain}
\markboth{J. KELLER, C. LASSER, AND T. OHSAWA}{A NEW PHASE SPACE DENSITY FOR QUANTUM EXPECTATIONS} 
\maketitle

\begin{abstract}
We introduce a new density for the representation of quantum states on phase space. It is constructed as a weighted difference of two smooth probability densities using the Husimi function and first-order Hermite spectrograms. In contrast to the Wigner function, it is accessible by sampling strategies for positive densities. In the semiclassical regime, the new density allows to approximate expectation values to second order with respect to the high frequency parameter and is thus more accurate than the uncorrected Husimi function. As an application, we combine the new phase space density with Egorov's theorem for the numerical simulation of time-evolved quantum expectations by an ensemble of classical trajectories. We present supporting numerical experiments in different settings and dimensions. 
\end{abstract}

\section{Introduction}

The wave functions describing the nuclear part of a molecular quantum system are square integrable functions on $\R^d$ with specific properties. They are smooth functions, but highly oscillatory and the dimension $d$ is large. 
The frequencies of oscillations are typically related to a small semiclassical parameter $\eps>0$, which can be thought of as the square root of the ratio of the electronic versus the average nuclear mass for the molecular system of interest. One expects
\begin{equation*}
\int_{\R^d} \overline{\psi}(x)\,(-\i\eps\nabla_x)\psi(x) \,\d x = O(1)
\end{equation*}
as $\eps\to0$ for most nuclear wave functions $\psi\in L^2(\R^d)$. Often the semiclassical analysis of a molecular quantum system requires a phase space representation of the nuclear wave function, the most popular being the Wigner function
\begin{equation*}
\W_\psi(z) \defeq (2\pi\eps)^{-d} \int_{\R^d} \overline\psi(q+\tfrac{y}{2}) \psi(q-\tfrac{y}{2}) \e^{\i y\cdot p/\eps} \d y,\qquad z=(q,p)\in\R^{2d}.
\end{equation*}
The Wigner function is a square integrable real-valued function on the phase space $T^*\R^d \cong \R^{2d}$ with many striking properties as for example
\begin{align*}
&\int_{\R^d} x|\psi(x)|^2 \,\d x = \int_{\R^{2d}} q \,\W_\psi(z)\, \d z,\\
& \int_{\R^d} \overline{\psi}(x)\,(-\i\eps\nabla_x)\psi(x) \,\d x = \int_{\R^{2d}} p\, \W_\psi(z) \,\d z.
\end{align*}
However, in most cases the Wigner function is not a probability density on phase space, since it may attain negative values. 

A guiding example is provided by the superposition of Gaussian wave packets. The semiclassically scaled Gaussian wave 
packet $g_{z_0}$ with phase space center $z_0=(q_0,p_0)\in\R^{2d}$ is defined as  
\begin{equation}\label{eq:gauss}
g_{z_0}(x) \defeq (\pi\eps)^{-d/4} \exp\!\left(-\tfrac{1}{2\eps}|x-q_0|^2 + \tfrac{\i}{\eps}p_0\cdot(x-\tfrac12 q_0)\right),\qquad x\in\R^d.
\end{equation}
It satisfies
\begin{equation*}
\int_{\R^d} x|g_{z_0}(x)|^2\,\d x = q_0\quad\text{and}\quad
\int_{\R^d} \overline{g_{z_0}}(x)\,(-\i\eps\nabla_x)g_{z_0}(x) \,\d x = p_0.
\end{equation*}
Its Wigner function is a nonnegative Gaussian function centered at the point $z_0$. However, the Wigner function of the superposition \begin{equation*}
\psi=g_{z_1}+g_{z_2},\qquad z_1,z_2\in\R^{2d},
\end{equation*} 
has three regions of localization as seen in the left panel of Figure~\ref{fig:superposition}. 
There are two regions around the points $z_1$ and $z_2$, respectively, where the Wigner function has nonnegative Gaussian shape, whereas in between around the midpoint of $z_1$ and $z_2$ there is an oscillatory region with pronounced negative values.  

\begin{figure}[h!]
\includegraphics[width=4cm]{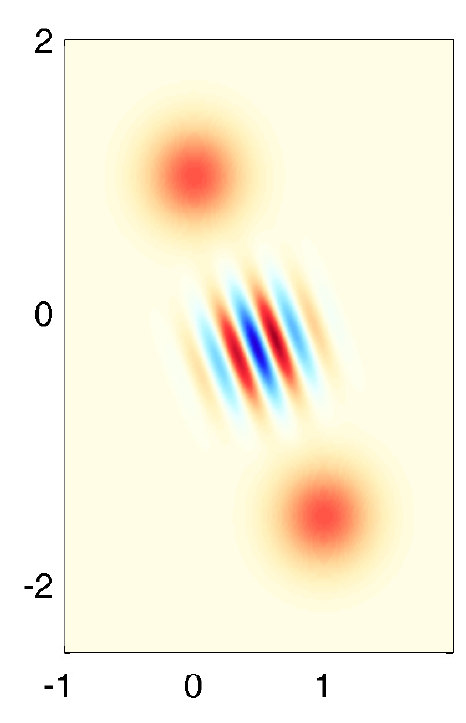} \includegraphics[width=4cm]{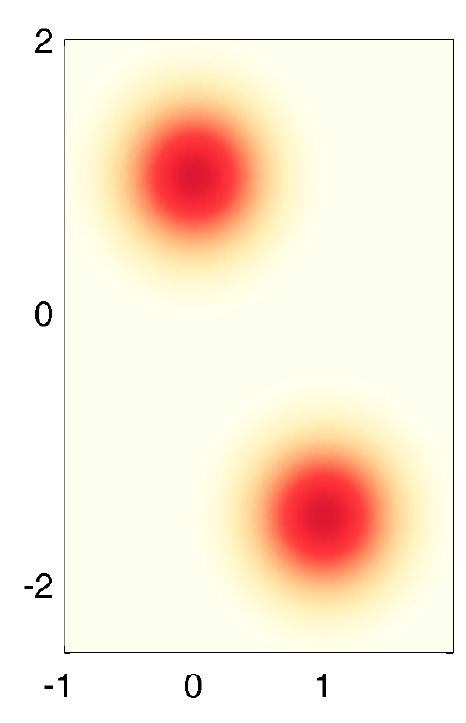} \includegraphics[width=4cm]{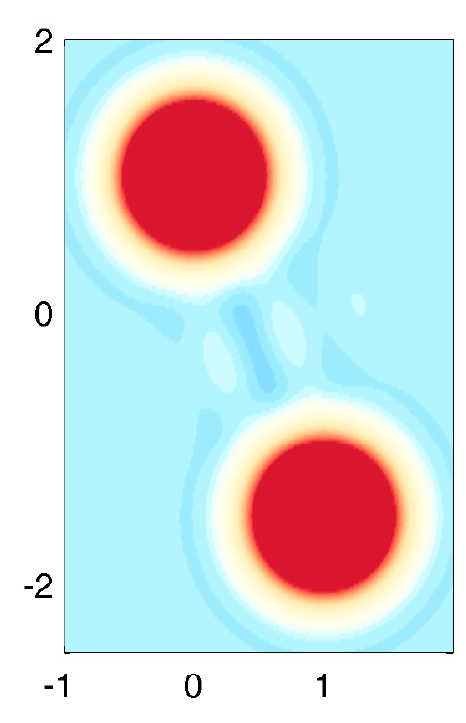}
\caption{Contour plots of the Wigner function (left), the Husimi function (middle), and the density $\mu_\psi$ (right) for a Gaussian superposition $\psi = g_{z_1} + g_{z_2}$.
We chose the phase space centers $z_1=(0,1)$, $z_2=(1,-\tfrac32)$, and the semiclassical parameter $\eps = 0.14$. Negative values are indicated by blue color.\label{fig:superposition}}
\end{figure}

One way of obtaining a nonnegative phase space representation of a wave function is to convolve its Wigner 
function with another Wigner function. One then calls  the nonnegative function
\begin{equation*}
\W_\psi*\W_\phi,\qquad \phi\in\Sc(\R^d),
\end{equation*}
a {\em spectrogram} of $\psi$. A widely used spectrogram is the Husimi function of $\psi$,
\begin{equation*}
\H_\psi \defeq \W_\psi*\W_{g_0},
\end{equation*} 
which is the spectrogram of $\psi$ with $\phi$ being the Gaussian wave packet~$g_{0}$ centered at the phase space origin.
For the superposition example, the smoothing of the convolution removes the oscillations and widens the 
Gaussian profiles around the centers $z_1$ and $z_2$; see the middle 
panel of Figure~\ref{fig:superposition}. However, the smoothing also destroys important 
exact relations satisfied by the Wigner function; {see also \cite[\S 7]{J97}}. Let $a:\R^d\to\R$ be a smooth function with an appropriate decay property. Then, 
\begin{equation*}
\int_{\R^d} a(x)|\psi(x)|^2 \,\d x = \int_{\R^{2d}} a(q) \W_\psi(z) \,\d z,
\end{equation*}
while
\begin{equation*}
\int_{\R^d} a(x)|\psi(x)|^2 \,\d x = \int_{\R^{2d}} a(q) \H_\psi(z) \,\d z + O(\eps)
\end{equation*}
as $\eps\to0$, where the error term depends on second and higher order derivatives of the function $a$. 
Hence only the first moment of the position density $x\mapsto|\psi(x)|^2$ is exactly recovered by the Husimi function.

Our aim is now to systematically construct a new phase space density that maintains some of the key properties of the Wigner function, while being amenable to sampling strategies for positive densities. We propose a proper reweighting of the Husimi function and the spectrograms obtained from the first order Hermite functions 
\begin{equation*}
\varphi_{e_j}(x) \defeq (\pi\eps)^{-d/4} \sqrt{\tfrac{2}{\eps}}x_j \exp\!\left(-\tfrac{1}{2\eps}|x|^2\right),\qquad x\in\R^d,\qquad j=1,\ldots,d.
\end{equation*}
We define a real-valued function $\mu_\psi:\R^{2d}\to\R$ by
\begin{equation*}
  \mu_\psi \defeq (1+\tfrac{d}2)\W_\psi*\W_{g_0}- \tfrac12 \sum_{j=1}^d \W_\psi*\W_{\varphi_{e_j}}. 
\end{equation*} 
By construction, $\mu_\psi$ is the difference of two nonnegative functions, a scalar multiple of the Husimi function  $\W_\psi*\W_{g_0}$ on the one side, and half the sum of the Hermite spectrograms 
$\W_\psi*\W_{\varphi_{e_1}},\ldots,\W_\psi*\W_{\varphi_{e_d}}$ on the other side. 
For example, a single Gaussian wave packet $\psi=g_{z_0}$ centered at the point $z_0\in\R^{2d}$ results in 
\begin{align*}
\mu_{g_{z_0}}(z) &= \left(1+\tfrac{d}{2}\right) (2\pi\eps)^{-d} \exp\!\left(-\tfrac{1}{2\eps}|z-z_0|^2\right)\\
  &\quad -\tfrac{1}{2\eps}|z-z_0|^2\, (2\pi\eps)^{-d} \exp\!\left(-\tfrac{1}{2\eps}|z-z_0|^2\right),
\end{align*}
which is the difference of two well-localized positive densities. For the superposition of two Gaussian wave packets 
$\psi = g_{z_1} + g_{z_2}$, we obtain a density $\mu_\psi$ characterized by two islands of positive values around the centers $z_1$ and $z_2$, which are surrounded by a sea of negative values; see the contour plot in the right panel of Figure~\ref{fig:superposition}
and the explicit formula in \S\ref{sec:superposition}. 

The new phase space function $\mu_\psi$ allows for the exact representation of the moments of $\psi$ up to order three in the following sense. 
If $a:\R^d\to\R$ is a polynomial of degree less than or equal to three, then
\begin{align*}
&\int_{\R^d} a(x) |\psi(x)|^2 \,\d x = \int_{\R^{2d}} a(q) \mu_\psi(z) \,\d z,\\
&\int_{\R^d} \overline{\psi}(x) a(-\i\eps\nabla)\psi(x) \,\d x = \int_{\R^{2d}} a(p) \mu_\psi(z) \,\d z.
\end{align*}
For arbitrary smooth functions $a:\R^{2d}\to\R$ and the associated Weyl quantized operator $\op(a)$, 
the quantum expectation value 
$\left\langle \psi,\op(a)\psi \right\rangle_{L^2}$ of the observable $\op(a)$
is approximated as
\begin{equation}\label{eq:second}
\int_{\R^d}\overline{\psi}(x)\, \op(a)\psi(x)\, \d x = \int_{\R^{2d}} a(z) \mu_\psi(z) \d z + O(\eps^2)
\end{equation}
for $\eps\to0$, where the error term depends on fourth and higher order derivatives of the function~$a$; 
see Theorem~\ref{thm:mu} later on. Phase space approximations of
quantum expectations and Wigner functions play a central role in the  analysis of quantum systems,
in particular in the semiclassical regime; 
see~\cite[\S IV]{LP93} or \cite[\S 7.1]{GMMP}.
The idea of combining different spectrograms can also be found  in the time-frequency literature; see e.g.~\cite{LPH94}
and the references given therein. 
However, the goal of~\cite{LPH94}
is cross-entropy optimization within a chosen set of spectrograms and not
the approximation of Wigner functions or quantum expectations.

The second order accuracy with respect to $\eps$ in the expectation value approximation suggests 
using the new density in the context of molecular quantum dynamics. We consider the time-dependent Schr\"odinger equation
\begin{equation*}
\i \eps \partial_t \psi(t)= \left(-\tfrac{\eps^2}{2}\Delta + V\right)\psi(t), \quad \psi(0) = \psi_0,
\end{equation*}
with a smooth potential function $V:\R^d\to\R$ as provided by the time-dependent Born--Oppenheimer 
approximation. 
Let us denote by $\Phi_t:\R^{2d}\to\R^{2d}$ the flow of the corresponding classical equations of motion
\begin{equation*}
\dot q = p,\qquad \dot p = -\nabla V(q).
\end{equation*}  
Then, by Egorov's theorem, we have 
\begin{equation}
\label{eq:egorov_wigner}
\left\langle \psi(t) ,\op(a)  \psi(t)\right\rangle_{L^2} =  
\int_\Rdd (a\circ\Phi_t)(z) \W_{\psi_0}(z) \d z+ O(\eps^2)
\end{equation}
as $\eps\to0$, where the error depends on third and higher order derivatives of the functions $a\circ\Phi_t$ and $V$ as well as the $L^2$-norm of the initial wave function~$\psi_0$. The Egorov approximation is computationally 
advantageous, in particular in high dimensions, since it allows to simulate the time-evolution of quantum expectations by an ensemble of classical trajectories. Over decades, it has been widely used in 
the physical chemistry literature under the name linearized semiclassical initial value representation (LSC-IVR) or Wigner phase space method.

Our new phase space density comes into play here, since the combination of the approximations in \eqref{eq:second} and \eqref{eq:egorov_wigner} provides
\begin{equation*}
\left\langle \psi(t) ,\op(a)  \psi(t)\right\rangle_{L^2} =  
\int_\Rdd (a\circ\Phi_t)(z) \mu_{\psi_0}(z) \d z+ O(\eps^2),
\end{equation*}
which can be read as a new method for the computation of time-evolved quantum expectations by initial sampling 
from a difference of nonnegative phase space distributions; see Theorem~\ref{thm:prop_husimi} and the numerical experiments in \S\ref{sec:numerics}.

\subsection{Outline}
Our investigation proceeds along the following lines. In \S\ref{sec:psd} we briefly review phase space distributions as the Wigner function, spectrograms, and the Husimi function. \S\ref{sec:new} derives the new phase space density $\mu_\psi$ and proves our main result, that is, the second order approximation of expectation values by the phase space integration with respect to the new density function. \S\ref{sec:quantum_dynamics} applies this result to the quantum propagation of 
expectation values. Then, several explicit formulas for the density $\mu_\psi$ are derived in \S\ref{sec:examples}. The numerical experiments in \S\ref{sec:numerics} illustrate the applicability of the new approach for the dynamics of molecular quantum systems in dimensions $d=1$, $d=2$, and $d=32$.
Appendix~\ref{sec:gamma_sampling} presents a sampling strategy for the density $\mu_\psi$ via the Gamma distribution used for the numerical experiments, 
while Appendix~\ref{app:num} provides further computational details.

\section{Phase space distributions}\label{sec:psd}
In this section we review different possibilities for representing a square integrable function $\psi\in L^2(\Rd)$ via real-valued distributions on the classical phase space 
$T^*\R^d \cong\Rdd$. Considering functions with frequencies of the order $1/\eps$ for a small parameter $0<\eps\ll1$, we work with the $\eps$-rescaled Fourier transform
$$
\mathcal{F}_\eps\psi(p) \defeq (2 \pi\eps)^{-d/2} \int_\Rd \psi(q) \e^{-\i p\cdot q/\eps}\d q,\qquad p\in\R^d.
$$

We also use the Heisenberg--Weyl operator in $\eps$-scaling:

\begin{definition}
 The {\em Heisenberg--Weyl operator} associated with a phase space point $z=(q,p)\in\Rdd$ is defined as 
\begin{equation*}
T_z \psi \defeq \e^{\i p\cdot(\bullet-q/2)/\eps}\psi(\bullet-q),\qquad\psi\in L^2(\R^d).
\end{equation*}
\end{definition}

Among its many striking properties, the following two will be important for us later on. We have 
\begin{equation*}
  T_z^\dagger = T_{-z},\qquad z\in\Rdd,
\end{equation*}
and
\begin{equation*}T_{z_{1}} T_{z_{2}} = \exp\parentheses{
    -\tfrac{\i}{2\eps} \Omega(z_{1}, z_{2})
  }\, T_{z_{1} + z_{2}},\qquad z_1,z_2\in\Rdd,
\end{equation*}
where $\Omega:\Rdd\times\Rdd\to\R$ denotes the standard symplectic form on $\Rdd$, i.e.,
\begin{equation}\label{eq:sympl_form}
  \Omega(z_{1}, z_{2}) \defeq z_{1}^{T} J z_{2} = q_{1}^T p_{2} - p_{1}^T q_{2}
  \qquad\text{with}
  \qquad
  J =
  \begin{bmatrix}
    0 & \Id \\
    -\Id & 0
  \end{bmatrix}.
\end{equation}

All phase space distributions considered here turn the action of the Heisenberg--Weyl operator $T_z$ on a wave function into a phase space translation by $z$, which is often referred to as a covariance property; see, e.g., \eqref{eq:spectrogram-translated} below.

\begin{remark}
  The Gaussian wave packet~\eqref{eq:gauss} with its phase space center at $z_{0} \in \R^{2d}$ is obtained by applying the Heisenberg--Weyl operator $T_{z_{0}}$ to the Gaussian
  \begin{equation*}
    g_{0}(x) \defeq (\pi\eps)^{-d/4} \exp\!\left(-\tfrac{1}{2\eps}|x|^2 \right)
  \end{equation*}
  centered at the origin, i.e., $g_{z_{0}} = T_{z_{0}} g_{0}$.
\end{remark}

\subsection{Wigner functions} We start our discussion with the celebrated Wigner function and recapitulate some basic relations.
\begin{definition}
The {\em Wigner function} of a function $\psi\in L^2(\R^d)$ is defined as $\W_\psi:\Rdd\to\R$, 
\begin{equation*}
\W_\psi(z) \defeq (2\pi\eps)^{-d} \int_{\R^d} \overline\psi(q+\tfrac{y}{2}) \psi(q-\tfrac{y}{2}) \e^{\i y\cdot p/\eps} \d y,\qquad z = (q, p) \in \R^{2d}.
\end{equation*}
\end{definition}

Wigner functions are continuous square-integrable functions on phase space; however, they need not be integrable. 
The marginals are the position and momentum density of the state, respectively. With a proper interpretation of the possibly not absolutely convergent integrals this means
$$
\int_\Rd  \W_\psi(q,p)\d p = |\psi(q)|^2, \qquad \int_\Rd  \W_\psi(q,p)\d q = |\mathcal{F}_\eps\psi(p)|^2,
$$
and in particular
\begin{equation*}
\int_{\R^{2d}} \W_\psi(z) \d z = \|\psi\|^2.
\end{equation*}
Wigner transformation preserves orthogonality in the sense that
\[
\int_{\Rdd} \W_\psi(z) \W_\phi(z) \d z = (2\pi\eps)^{-d} \left|\langle\psi,\phi\rangle\right|,\qquad \psi,\phi\in L^2(\R^d),
\]
and it turns the action of the Heisenberg--Weyl operator into a phase space translation, i.e.,
\begin{equation*}
  \W_{T_z\psi} = \W_\psi(\bullet - z),\qquad z\in\R^{2d},
\end{equation*}
which is an example of the covariance property alluded above.

Moreover, given a Schwartz function $a:\Rdd \to \R$, one can use Wigner functions to express expectation values of
Weyl quantized linear operators
\begin{equation}\label{eq:weyl_quant}
(\op(a) \psi)(q) = (2\pi \eps)^{-d} \int_\Rdd a(\tfrac{q+y}2,p) \psi(y) \e^{\i(q-y)\cdot p/ \eps} \d y\,\d p
\end{equation}
via the weighted phase space integral
\begin{equation}\label{eq:weyl-wigner}
\left\langle \psi , \op(a) \psi\right\rangle = \int_\Rdd a(z) \W_\psi(z) \d z.
\end{equation}
We note that the oscillatory integral formula~(\ref{eq:weyl_quant}) can be extended to more general classes of symbols $a:\Rdd \to \R$
with controlled growth properties at infinity; see for example~\cite[\S4]{Z12}, \cite[\S2]{M02} or \cite[\S2]{F89}. 

\subsection{Spectrograms}
Except for Gaussian states, Wigner functions attain negative values (see \cite{SC83}), and thus cannot be treated as
probability densities. For example, any odd function $\psi\in L^2(\R^d)$ satisfies
$$
\W_\psi(0) = -(2\pi\eps)^{-d}\int_{\R^d}|\psi(\tfrac{y}{2})|^2\, \d y \le 0.
$$
One way to obtain  nonnegative phase space representations of a quantum state  is to convolve its Wigner function
with another Wigner function.

\begin{definition}
Let $\psi \in L^2(\Rd)$ and $\phi \in \Sc(\Rd)$. Then, $\W_\psi * \W_\phi$ is called a \emph{spectrogram} of $\psi$.
\end{definition}

In time-frequency analysis, spectrograms are typically introduced as the modulus squared of a short-time Fourier transform (see, e.g., the introduction in~\cite{Flandrin15}) so that the representation via the convolution of two Wigner transforms is derived subsequently. Spectrograms also form a sub-class of Cohen's class of phase space distributions \cite[\S3.2.1.]{F99}. They satisfy
\begin{equation}\label{eq:spec_square}
(\W_\psi*\W_{\phi})(z) = (2\pi\eps)^{-d} \left|\langle T_z\phi_-,\psi\rangle\right|^2,\qquad z\in\R^{2d},
\end{equation}
where $\phi_-(x) \defeq \phi(-x)$ for $x\in\R^d$; see also \cite[Proposition 1.99]{F89}. Thus, spectrograms are nonnegative  and smooth by construction. 
The integrability follows from~\eqref{eq:spec_square} by 
the square integrability of general Fourier-Wigner transforms $z \mapsto \langle T_z\phi,\psi\rangle$ with $\phi,\psi\in L^2(\Rd)$, see Proposition 1.42  in \cite{F89}.
Normalization is preserved according to 
\begin{equation}
  \label{eq:integral_of_spectrogram}
  \int_{\R^{2d}} (\W_\psi*\W_\phi)(z) \d z = \|\psi\|^2 \cdot \|\phi\|^2.
\end{equation}
Spectrograms also inherit the covariance property from the Wigner function, i.e.,
\begin{equation}
\label{eq:spectrogram-translated}
\W_{T_z\psi}*\W_{\phi} = (\W_\psi*\W_{\phi})(\bullet - z),\qquad z\in\R^{2d}.
\end{equation}
A particular spectrogram is obtained by convolving with the Wigner function of a Gaussian wave packet, which we will discuss next.

\subsection{Husimi functions}
The most commonly used nonnegative phase space 
distribution is the Husimi function; see e.g.~\cite[\S4.1]{AMP09}. 
We consider the Wigner function
\begin{equation*}
  \W_{g_{0}}(z)=(\pi\eps)^{-d}\e^{-|z|^2/\eps}, \qquad z\in\Rdd,
\end{equation*}
of the Gaussian wave packet $g_{0}$ centered at the phase space origin and define:

\begin{definition}
The {\em Husimi function} of $\psi\in L^2(\R^d)$ is defined as the spectrogram
\begin{equation*}
  \H_\psi(z) \defeq (\W_\psi * \W_{g_0})(z) = \int_{\Rdd} \W_{\psi}(w)\,(\pi\eps)^{-d}\e^{-|z-w|^2/\eps}\,\d{w}.
\end{equation*}
\end{definition}

The Husimi function of $\psi$ is the spectrogram \eqref{eq:spec_square} with $\phi$ being the Gaussian wave packet $g_{0}$, and since $g_{0}$ is even, we have
$$
\H_\psi(z) =  (2\pi\eps)^{-d} \left| \langle T_z g_0,\psi\rangle\right|^2,\qquad z\in\Rdd.
$$
Also, since $g_{0}$ is normalized, i.e., $\|g_0\|=1$, \eqref{eq:integral_of_spectrogram} gives
$$
\int_\Rdd \H_\psi(z) \d z = \|\psi\|^2.
$$
That is, for $\psi\in L^2(\Rd)$ with $\|\psi\|=1$, the Husimi function is a smooth probability density on phase space. 

\begin{remark}
The Husimi function of $\psi\in L^2(\R^d)$ is the modulus squared of the so-called Fourier--Bros--Iagolnitzer (FBI) transform, which associates with $\psi$ the mapping
\[
\Rdd\to\C,\quad z\mapsto (2\pi\eps)^{-d/2} \langle T_z g_0,\psi\rangle.
\]
In contrast to the Wigner function and the spectrograms, the FBI transform is a linear, albeit complex-valued phase space representation; see \cite[Chapter~3]{M02}. 
\end{remark}

Integrating a Schwartz function $a:\Rdd\to\R$ against the Husimi function, we obtain
\begin{align*}
\int_\Rdd a(z) \H_\psi(z) \d z &= \int_\Rdd a(z) (\W_\psi*\W_{g_0})(z) \d z
= \int_\Rdd (a*\W_{g_0})(z) \W_\psi(z) \d z\\
&= \left\langle\psi,\op(a*\W_{g_0})\psi\right\rangle,
\end{align*}
where the last equation uses the duality relation~(\ref{eq:weyl-wigner}) between Weyl quantized operators and the Wigner transform. Therefore,
\begin{equation*}
  \int_\Rdd a(z) \H_\psi(z) \d z = \langle\psi,\op_{\rm aw}(a)\psi\rangle,
\end{equation*}
where 
$$
\op_{\rm aw}(a) \defeq \op(a*\W_{g_0}),\qquad a\in\Sc(\Rdd),
$$
denotes the anti-Wick quantized operator of the function $a$; see for example \cite[\S2.7]{F89} and \cite[\S11.4]{G11}. Weyl and anti-Wick quantization are $\eps$-close in the following sense:

\begin{lemma}\label{lem:star_to_weyl}
Let $a:\Rdd\to \R$ be a Schwartz function and $\eps>0$. Then, there are two families of Schwartz functions $r^\eps_1,r^\eps_2:\Rdd\to \R$ that depend on fourth and higher order derivatives of $a$, so that
\begin{align*}
\op_{\rm aw}(a) &= \op(a + \tfrac\eps4 \Delta a) + \eps^2\op(r^\eps_1),\\
\op_{\rm aw}(a - \tfrac\eps4 \Delta a) &= \op(a) + \eps^2 \op(r^\eps_2),
\end{align*}
where $\sup_{\eps>0}\| \op(r^\eps_j)\|_{L^2\to L^2}<\infty$ for both $j=1,2$. 
\end{lemma}

\begin{proof}
The lemma is essentially proven in~\cite[Proposition~2.4.3]{L10} or \cite[Lemma 1]{KL13}, and hence we only sketch
the proof for the second of the two equivalent identities. We write out the definition
\begin{align*}
\op_{\rm aw}(a-\tfrac{\eps}{4}\Delta a) &= \op\left(\W_{g_0} * (a-\tfrac\eps4\Delta a)\right)
\end{align*}
and Taylor expand $a-\tfrac\eps4\Delta a$ around $z$ in the integral
\[
\W_{g_0} * (a-\tfrac\eps4\Delta a) = (\pi \eps)^{-d} \int_\Rdd (a-\tfrac\eps4\Delta a)(\zeta) \e^{-|z-\zeta|^2/\eps} \d\zeta.
\]
Due to the symmetry of the Gaussian, all Taylor expansion terms with odd derivatives of  $(a-\tfrac\eps4\Delta a)$ vanish.
The computation
\[
(\pi \eps)^{-d}  \sum_{|\alpha|=1}\int_\Rdd \frac1{(2\alpha)!} (\partial^{2\alpha}(a-\tfrac\eps4\Delta a)) (\zeta - z)^{2\alpha} \e^{-|z-\zeta|^2/\eps} \d\zeta = \frac{\eps}{4} \Delta(a-\tfrac\eps4\Delta a)
\]
implies the second order approximation
\begin{align*}
\W_{g_0} * (a-\tfrac\eps4\Delta a) &= (a-\tfrac\eps4\Delta a) + \tfrac\eps{4} \Delta (a-\tfrac\eps4\Delta a)  + O(\eps^2)\\
&= a + O(\eps^2),
\end{align*}
where the $O(\eps^2)$ term
is of Schwartz class. Applying the Calder\'on--Vaillancourt Theorem (see, e.g., \cite[\S2.5]{F89})
concludes the proof.
\end{proof}

\begin{remark}
The result of Lemma~\ref{lem:star_to_weyl} can formally be read in terms of the heat semigroup $\{\exp( t \Delta)\}_{t\geq 0}$ as 
\begin{equation*}
a*\W_{g_0} = \exp( \tfrac{\eps}4  \Delta) a = a + \tfrac\eps4 \Delta a + O(\eps^2),
\end{equation*}
where $a:\Rdd\to\R$ is a Schwartz function.
\end{remark}

\section{The new phase space density}\label{sec:new}

We learn from the preceding discussion of phase space distributions, in particular from Lemma~\ref{lem:star_to_weyl}, that the Husimi function allows to approximate an expectation value according to
\begin{equation*}
  \left\langle \psi , \op(a) \psi\right\rangle = \int_\Rdd (a-\tfrac\eps4\Delta a)(z) \H_\psi(z) \d z + O(\eps^2)
\end{equation*}
as $\eps\to0$, where the error depends on the fourth and higher order derivatives of $a$ and the $L^2$-norm of $\psi$. 
An integration by parts provides
\begin{equation*}
  \left\langle \psi , \op(a) \psi\right\rangle = \int_\Rdd a(z)\,(\H_\psi -\tfrac\eps4\Delta\H_{\psi})(z)\,\d z + O(\eps^2),
\end{equation*}
and motivates us to define the following new phase space density.

\begin{definition}\label{def:mu}
For $\psi \in L^2(\Rd)$ we define the phase space density $\mu_\psi:\Rdd \to \R$, 
\begin{align*}
  \mu_\psi &\defeq  \H_\psi - \tfrac\eps4 \Delta \H_\psi.
\end{align*}
\end{definition}

We summarize the key property of the new density as follows: 

\begin{theorem}
\label{thm:mu}
Let $a:\Rdd\to\R$ be a Schwartz function. Then, there exists a constant $C \ge0$ depending on fourth and higher order derivatives of $a$ such that for all $\psi\in L^2(\Rd)$ 
$$
\left| \langle\psi,\op(a)\psi\rangle - \int_\Rdd a(z) \mu_\psi(z) \d z \right| \le C\eps^2 \|\psi\|^2,
$$
where the density $\mu_\psi:\Rdd \to \R$ was defined in Definition~\ref{def:mu}.
\end{theorem}

\begin{proof}
By Lemma~\ref{lem:star_to_weyl}, we have
$$
\langle\psi,\op(a)\psi\rangle = \int_\Rdd (a-\tfrac\eps4\Delta a)(z) \H_\psi(z) \d z + \eps^2 \langle\psi,\op(r^\eps)\psi\rangle
$$
with
$$
|\langle\psi,\op(r^\eps)\psi\rangle| \;\le\; C \|\psi\|^2,
$$
where the constant $C>0$ depends on the fourth and higher order derivatives of $a$.
The Husimi function $\H_\psi$ is smooth and bounded since
  \begin{equation*}
    \H_\psi(z) =  (2\pi\eps)^{-d} \left| \langle g_z,\psi\rangle\right|^2 
    \le (2\pi\eps)^{-d} \norm{g_{z}}^{2} \norm{\psi}^{2}
    = (2\pi\eps)^{-d} \norm{\psi}^{2}.
  \end{equation*}
Hence, integration by parts implies
$$
\int_\Rdd (a-\tfrac\eps4\Delta a)(z) \H_\psi(z) \d z = \int_\Rdd a(z) (\H_\psi-\tfrac\eps4\Delta\H_\psi)(z) \d z = \int_\Rdd a(z) \mu_\psi(z) \d z.
$$
Therefore, 
$$
\left|\langle\psi,\op(a)\psi\rangle - \int_\Rdd a(z) \mu_\psi(z) \d z\right| = \eps^2| \langle\psi,\op(r^\eps)\psi\rangle| \;\le\; C\,\eps^2\|\psi\|^2,
$$
which concludes the proof.
\end{proof}

\begin{remark}
Whenever $a:\Rdd \to \R$ is a polynomial of degree less than or equal to three,
the constant $C\ge0$ of Theorem~\ref{thm:mu} vanishes so that the phase space integration with respect to 
$\mu_\psi$ exactly reproduces the expectation value. Moreover, using higher order Hermite spectrograms, one can also construct
densities which yield approximations of expectation values with higher order errors in~$\eps$, or, equivalently,
which are exact for polynomial symbols $a:\Rdd \to \R$ of higher degree. We refer to the thesis~\cite[\S10.5]{K15}
for the next order result and an outline on how to prove higher order approximations.
\end{remark}

Our next aim is to derive an alternative expression for the new density showing that it is a linear combination of spectrograms. 

\subsection{The new density in terms of Hermite functions}

The Laplacian of the Husimi function can be related to the Wigner function of the $\eps$-rescaled first order Hermite functions as follows: 
\begin{proposition}\label{lem:laplace} Let $\eps>0$, $j \in \{1,\ldots,d \} $, and
$$
\varphi_{e_j}(x) \defeq (\pi\eps)^{-d/4} \sqrt{\tfrac{2}{\eps}}x_j \exp\!\left(-\tfrac{1}{2\eps}|x|^2\right),\qquad x\in\R^d,
$$
be the first order Hermite functions. Then, for all $\psi\in L^2(\Rd)$,
$$
\Delta\H_\psi = \tfrac2\eps \sum_{j=1}^d \W_\psi*\W_{\varphi_{e_j}} - \tfrac{2d}{\eps}\H_{\psi}
$$
and consequently
\begin{equation*}
\mu_\psi = \H_\psi - \tfrac\eps4 \Delta \H_\psi = (1+\tfrac{d}{2})\H_\psi - \tfrac12\sum_{j=1}^d \W_{\psi}*\W_{\varphi_{e_j}}.
\end{equation*}
\end{proposition}

\begin{proof}
Let $z\in\R^{2d}$. We compute
\begin{align*}
\Delta \W_{g_0}(z) &= (\pi\eps)^{-d} \Delta \e^{-|z|^2/\eps} = (\pi\eps)^{-d}\, \nabla\cdot\left(-\tfrac2\eps z \,\e^{-|z|^2/\eps}\right)\\
&= (\pi\eps)^{-d} \left(-\tfrac{4d}{\eps} + \tfrac{4}{\eps^2}|z|^2\right) \e^{-|z|^2/\eps}.
\end{align*}
Moreover, by direct computation or \cite[Theorem~1]{LT14},
\begin{equation*}
  \W_{\varphi_{e_j}}(z) = -(\pi\eps)^{-d} \left(1-\tfrac{2}{\eps}|z_j|^2\right)\e^{-|z|^2/\eps} ,
\end{equation*}
such that
\begin{align*}
\sum_{j=1}^d \W_{\varphi_{e_j}}(z) &=
-(\pi\eps)^{-d} \left (2d-d-\tfrac{2}{\eps}|z|^2\right) \e^{-|z|^2/\eps} \\
&=
-(\pi\eps)^{-d}  \left (2d-\tfrac{2}{\eps}|z|^2\right)  \e^{-|z|^2/\eps}+ d \cdot \W_{g_0}(z)
\end{align*}
and
$$
\Delta \W_{g_0}(z) = -\tfrac2\eps (\pi\eps)^{-d} \left(2d - \tfrac{2}{\eps}|z|^2\right) \e^{-|z|^2/\eps} = \tfrac2\eps\sum_{j=1}^d \W_{\varphi_{e_j}}(z) - \tfrac{2d}{\eps}\,\W_{g_0}(z).
$$
To conclude the proof we note that $\Delta \H_\psi = \Delta (\W_\psi*\W_{g_0}) = \W_\psi*\Delta \W_{g_0}$.
\end{proof}

\begin{remark}
For any $\psi\in L^2(\R^d)$, the real-valued density $\mu_\psi$ is a weighted sum of spectrograms and therefore smooth and integrable. It satisfies the normalization condition
\[
\int_{\Rdd} \mu_\psi(z) \d z = \|\psi\|^2
\]
and the covariance property
\[
\mu_{T_z\psi} = \mu_\psi(\bullet-z),\qquad z\in\R^{2d}.
\]
\end{remark}

Next, we add a further characterization of the new density that does not explicitly require a convolution. 

\subsection{The new density in terms of ladder operators}
The Gaussian wave packet $g_0$ centered at the origin and the first order Hermite functions $\varphi_{e_1},\ldots,\varphi_{e_d}$ can be characterized by the raising and lowering operators
\begin{equation*}
A^\dagger = \tfrac{1}{\sqrt{2\eps}}(x-\eps\nabla_x)\quad\text{and}\quad 
A = \tfrac{1}{\sqrt{2\eps}}(x+\eps\nabla_x),
\end{equation*}
respectively. On the one hand, we have 
\[
{\rm span}\{g_0\} = \left\{\psi\in L^2(\R^d)\mid A_j\psi = 0 \;\text{for all}\;j=1,\ldots,d\right\},
\] 
for the kernel of the lowering operator $A=(A_1,\ldots,A_d)$. On the other hand, the components of the raising operator applied to the Gaussian wave packet $g_0$ generate the first order Hermite functions in the sense that
\[
\varphi_{e_j} = A_j^\dagger g_0,\qquad j=1,\ldots,d.
\]

\begin{proposition}\label{lem:FBI} 
For all $\psi\in L^2(\R^d)$, $j=1,\ldots,d$, and $z=(q,p)\in \R^{2d}$, we have
\begin{align*}
  (\W_\psi*\W_{\varphi_{e_j}})(z)
  &= (2\pi\eps)^{-d} \; | \langle T_z A_j^\dagger g_0,\psi\rangle|^2\\
  & = (2\pi\eps)^{-d} \left| \left\langle g_z,\left(A_j-\tfrac{1}{\sqrt{2\eps}}z_j^\C\right)\psi\right\rangle\right|^2,
\end{align*}
where $z^\C\defeq q+\i p\in\C^d$ and $g_z$ is the Gaussian wave packet defined in \eqref{eq:gauss}. Consequently, 
\begin{align*}
\mu_\psi(z)
&= (2\pi\eps)^{-d} \left(   (1+\tfrac{d}2) \big| \langle T_z g_0, \psi \rangle \big|^2 -  \tfrac12  
\sum_{j=1}^d  \big| \langle T_z A_j^\dagger g_0, \psi \rangle \big|^2  \right)\\
&=  (2\pi\eps)^{-d} \left(   (1+\tfrac{d}2) \big| \langle g_z, \psi \rangle \big|^2 -  \tfrac12  \sum_{j=1}^d  \left| \left\langle g_z,\left(A_j-\tfrac{1}{\sqrt{2\eps}}z_j^\C\right)\psi\right\rangle\right|^2  \right).
\end{align*}
\end{proposition}

\begin{proof}
The relation \eqref{eq:spec_square} implies
\begin{align*}
(\W_\psi*\W_{\varphi_{e_j}})(z) &= (2\pi\eps)^{-d} \; | \langle T_z (\varphi_{e_j})_-,\psi\rangle|^2\\
& = (2\pi\eps)^{-d} \; | \langle T_z A_j^\dagger g_0,\psi\rangle|^2,
\end{align*}
since $(\varphi_{e_j})_-(x) = \varphi_{e_j}(-x) = -\varphi_{e_j}(x)$ for $x\in\R^d$. For the second representation of the Hermite spectrogram 
we compute 
\[
\eps\partial_j\circ T_z = \i p_j T_z  +T_z \circ\eps\partial_j
\] 
and deduce
\begin{align*}
T_z\circ A^\dagger_j &=
\tfrac{1}{\sqrt{2\eps}}\left( (x_j-q_j) T_z - T_z\circ\eps\partial_j\right)\\
&= \tfrac{1}{\sqrt{2\eps}}\left((x_j-q_j)T_z+\i p_j T_z-\eps\partial_j\circ T_z\right)\\
&= \left(A_j^\dagger -\tfrac{1}{\sqrt{2\eps}}(q_j-\i p_j) \right)\circ T_z
\end{align*}
so that
\[
\langle T_z A_j^\dagger g_0,\psi\rangle = \langle A_j^\dagger g_z,\psi\rangle - \tfrac{1}{\sqrt{2\eps}}(q_j+\i p_j)\langle g_z,\psi\rangle. 
\]
\end{proof}

Proposition~\ref{lem:FBI} will be used for explicit expressions of $\mu_\psi$  later on in \S\ref{sec:examples}, when $\psi$ 
is a superposition of Gaussian wave packets or a higher order Hermite function.

\section{Quantum dynamics}
\label{sec:quantum_dynamics}

As an application of the new density we consider the approximation of expectation values for the solution of the time-dependent semiclassical Schr\"odinger equation
\begin{equation*}
\i\eps\partial_t\psi(t) = H \psi(t),\qquad \psi(0) = \psi_0,
\end{equation*}
where the Schr\"odinger operator $H= \op(h)$ is the Weyl quantization of a smooth function $h:\R^{2d}\to\R$ of subquadratic growth, that is, all derivatives of the
function~$h$ of order two and higher are bounded. Then, $H$ is essentially self-adjoint (see~\cite[Exercise IV.12]{R87}) so that for all square integrable initial data $\psi_0\in L^2(\R^d)$  there is a unique global solution
$$
\psi(t) = \e^{-\i Ht/\eps}\psi_0,\qquad t\in\R.
$$ 
The classical counterpart to the Schr\"odinger equation is the Hamiltonian ordinary differential equation
\begin{equation*}
\dot z(t)= J\nabla h(z(t)) \quad\text{with}\quad  J =
  \begin{bmatrix}
    0 & \Id \\
    -\Id & 0
  \end{bmatrix}\in\R^{2d\times 2d}.
\end{equation*}
The associated Hamiltonian flow $\Phi_t:\R^{2d}\to\R^{2d}$ is globally defined and smooth for all times $t\in\R$, since $h$ is smooth and subquadratic.

In this setup we obtain the following quasiclassical approximation of time evolved quantum expectations using the new phase space density. 

\begin{corollary}\label{thm:prop_husimi}
Suppose $h:\R^{2d}\to\R$ is a smooth function of subquadratic growth and $H=\op(h)$. Let $\psi \in L^2(\Rd)$ with $\| \psi \|_{L^2} = 1$.
Then, for all Schwartz functions $a:\Rdd\to\R$, and $t\in\R$, there exists a constant $C=C(a,h,t)\ge0$ such that
$$
\left| \left\langle \e^{-\i Ht/\eps}\psi,\op(a) \e^{-\i Ht/\eps}\psi\right\rangle - \int_\Rdd (a\circ\Phi_t)(z) \mu_\psi(z)\d z \right| \le C\eps^2
$$
with the density $\mu_\psi$ from Definition~\ref{def:mu}, where $\Phi_t:\Rdd\to\Rdd$ is the Hamiltonian flow associated with $h$.
\end{corollary}

\begin{proof} 
The crucial element of our argument is Egorov's theorem~\cite[Theorem 1.2]{BR02}, which provides
\[
\e^{\i Ht/\eps}\op(a)\e^{-\i Ht/\eps} = \op(a\circ\Phi_t) + O(\eps^2),
\]
where the error depends on third and higher order derivatives of $a\circ\Phi_t$ and $h$, respectively. This means for the expectation value
\begin{align*}
\left\langle \e^{-\i Ht/\eps}\psi,\op(a) \e^{-\i Ht/\eps}\psi\right\rangle &=
\Big\langle\psi,\op(a\circ\Phi_t) \psi\Big\rangle + O(\eps^2).
\end{align*}
Now it remains to apply Theorem~\ref{thm:mu} to obtain
\begin{equation*}
\left\langle \e^{-\i Ht/\eps}\psi,\op(a) \e^{-\i Ht/\eps}\psi\right\rangle = \int_{\R^{2d}} (a\circ\Phi_t)(z) \mu_\psi(z) \d z + O(\eps^2). 
\end{equation*}
\end{proof}

%
\begin{remark}\label{rem:naive_husimi}
Replacing the new density by the Husimi function in Corollary~\ref{thm:prop_husimi} deteriorates the approximation in the sense that
\begin{equation*}
 \left\langle \e^{- \i Ht/\eps}\psi,\op(a)\e^{-\i Ht/\eps}\psi\right\rangle =  \int_\Rdd (a\circ\Phi_t)(z) \H_\psi(z) \d z+ O(\eps)
\end{equation*}
as $\eps\to0$. It requires an additional system of coupled ODEs involving higher order derivatives of the Hamilton function $h$ to retain second order accuracy with respect to $\eps$; see~\cite[Theorem 3]{KL13}. 
\end{remark}

\begin{remark}\label{rem:tot_energy_error}
Since the Hamiltonian $h$ is preserved by the classical flow $\Phi_t$,
the constant $C(h,h,t) = C(h,h)$ of Corollary~\ref{thm:prop_husimi} does not depend on time, so that the approximation 
error of the total energy expectation value is of size $O(\eps^2)$ but time-independent.
\end{remark}

\begin{remark}
In the special case of a harmonic oscillator $h(z)=z^TAz$, with $A\in \R^{2d \times 2d}$ positive definite, generating a 
flow $\Phi_t$ that is a linear orthogonal map on phase space, one can easily
see that
\[
\mu_{\psi(t)} = \mu_{\psi_0} \circ \Phi_{-t} .
\]
In other words, $\mu_{\psi(t)}$ satisfies the classical Liouville equation. In general, the time evolution of 
$\mu_{\psi(t)}$ is much more intricate, see the equations for the Husimi function 
derived in~\cite[Theorem 4.5 and \S 4.2]{AMP09}.
\end{remark}

\section{Examples of Phase Space Densities}
\label{sec:examples}
In this section we explicitly compute the new density $\mu_\psi$  from Definition \ref{def:mu} in three different cases, 
namely when $\psi$ is a Gaussian wave packet,
a Gaussian superposition, or a multivariate Hermite function.

\subsection{Gaussian wave packets}\label{sec:gaussian}
The Gaussian wave packet
\begin{equation*}
g_z(x) = (\pi\eps)^{-d/4} \exp\!\left(-\tfrac{1}{2\eps}|x-q|^2 + \tfrac{\i}{\eps}p\cdot (x-\tfrac12 q)\right),
\quad x\in\R^d.
\end{equation*}
centered at $z=(q,p)\in\Rdd$ has the Wigner function
\begin{equation}
  \label{eq:Wigner_for_Gaussian}
\W_{g_z}(w) = (\pi\eps)^{-d} \exp\!\left(-\tfrac1\eps|w-z|^2\right), \qquad w\in\Rdd.
\end{equation}
Its Husimi function is a Gaussian function, too, but broader, that is,  
$$
\H_{g_z}(w) = (2\pi\eps)^{-d}  \exp\!\left(-\tfrac{1}{2\eps}|w-z|^2\right),\qquad w\in\Rdd.
$$
The Hermite spectrograms of $g_z$ can be related to another Husimi function. Indeed,   
by the covariance property of the spectrograms, 
\begin{align*}
(\W_{g_z}* \W_{\varphi_{e_j}})(w) &=
(\W_{g_0}*\W_{\varphi_{e_j}})(w - z) = \H_{\varphi_{e_j}}(w-z)\\
&= (2\pi\eps)^{-d} \tfrac{1}{2\eps}|w_j-z_j|^2 \exp\!\left(-\tfrac{1}{2\eps}|w-z|^2\right)
\end{align*}
for all $w=(w_1,\ldots,w_d)\in\R^{2d}$ with $w_1,\ldots,w_d\in \R^{2}$ and all $j=1,\ldots,d$.
Summing all the Hermite spectrograms then yields
$$
\mu_{g_z}(w) =  (2\pi\eps)^{-d}\left(1+\tfrac{d}{2}-\tfrac{1}{4\eps}|w-z|^2\right) \exp\!\left(-\tfrac{1}{2\eps}|w-z|^2\right),\qquad w\in\Rdd.
$$
Figure~\ref{fig:gaussian} plots the new density $\mu_{g_z}(w)$ together with the corresponding Wigner and Husimi function in terms of the distance $|w-z|$. Its polynomial prefactor puts the density $\mu_{g_z}$ in between the Wigner and the Husimi Gaussian.

\begin{figure}[ht!]
\includegraphics[width=\textwidth]{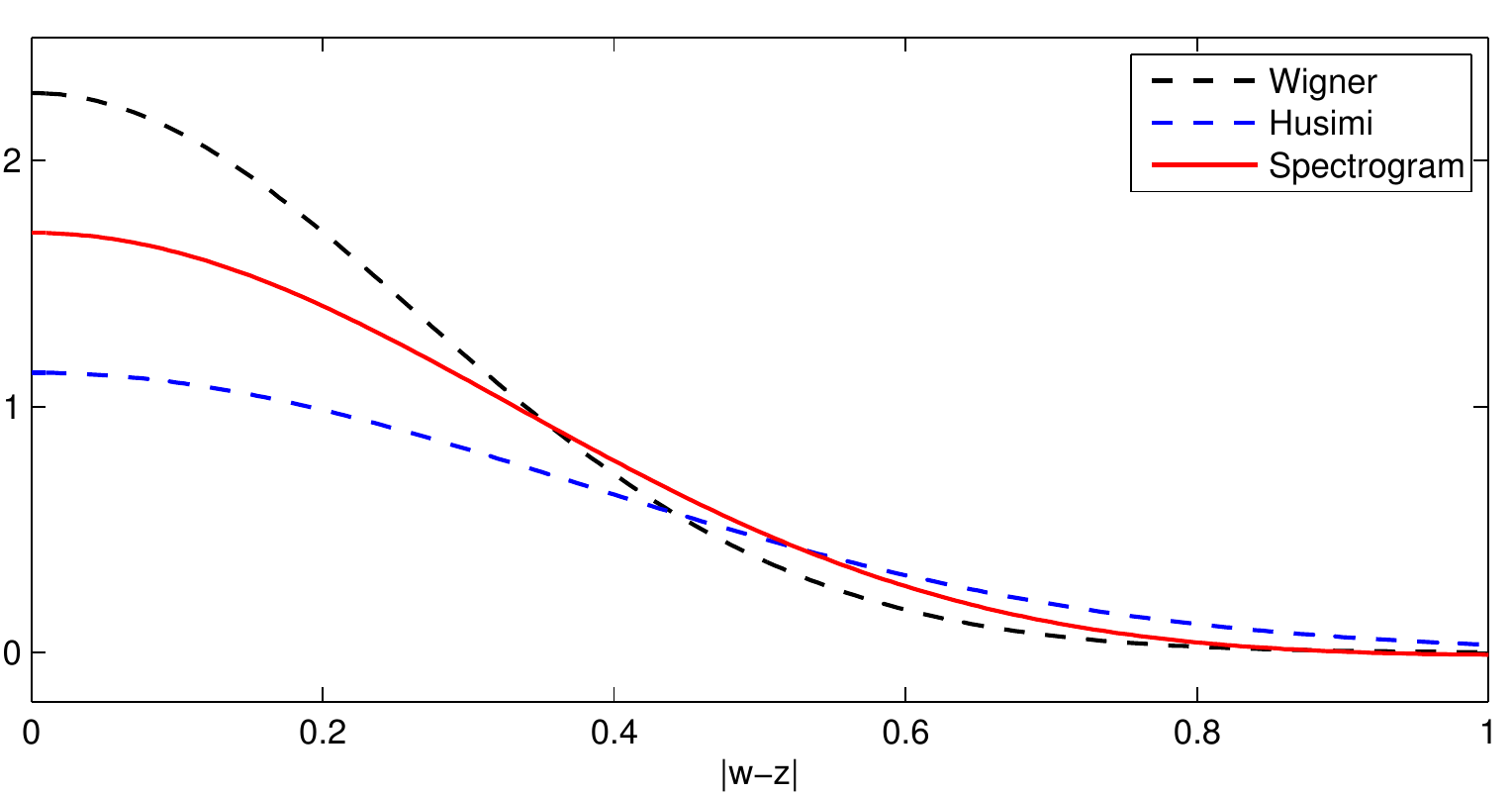}
\caption{\label{fig:gaussian}
Plots of the Wigner function (black dashed), the Husimi function (blue dashed), and the new density $\mu_{\psi}$ (red solid) for 
Gaussian wave packet $\psi=g_z$ in dimension $d=1$ with semiclassical parameter $\eps=0.14$. We plot the three functions in terms of the distance 
$|w-z|$.}
\end{figure}

\subsection{Gaussian superposition}\label{sec:superposition}
Next we compute the new density $\mu_\psi$ for a Gaussian superposition, that is, for  
\begin{equation*}
\psi = g_{z_1} + g_{z_2},\qquad z_1, z_2\in\Rdd.
\end{equation*}
Writing the value of the Husimi function at a point $z\in\Rdd$ as 
\begin{align*}
   \H_\psi(z)
  &= (2\pi\eps)^{-d} \left|\langle g_z, g_{z_{1}}\rangle + |\langle g_z, g_{z_{2}}\rangle\right|^{2} \\
  &= (2\pi\eps)^{-d} \;\left(
    |\langle g_z, g_{z_{1}}\rangle|^{2}
    + |\langle g_z, g_{z_{2}}\rangle|^{2}
    + 2\Re\parentheses{ \langle g_{z_{1}}, g_{z}\rangle \langle g_z, g_{z_{2}}\rangle }
  \right)
\end{align*}
motivates us to derive an explicit formula for the inner product of two Gaussian wave packets. 

\begin{lemma}\label{lem:ip}
  For any $z_{1}, z_{2} \in \R^{2d}$, we have
  \begin{equation*}
    \ip{g_{z_{1}}}{g_{z_{2}}} = \exp\parentheses{
        -\tfrac{|z_{1} - z_{2}|^{2}}{4\eps}
        + \tfrac{\i}{2\eps} \Omega(z_{1}, z_{2})
    },
  \end{equation*}
  where $\Omega$ is the standard symplectic form on $T^{*}\R^{d} \cong \R^{2d}$ as defined in \eqref{eq:sympl_form}.
\end{lemma}

\begin{proof}
By a direct calculation one can show 
 \begin{equation*}
    \ip{ g_{0} }{ g_{z} } = \ip{ g_{0} }{ T_{z} g_{0} }
    = \exp\parentheses{ -\tfrac{|z|^{2}}{4\eps} },\qquad z\in\Rdd.
  \end{equation*}
For the general case we then have
  \begin{align*}
    \ip{g_{z_{1}}}{g_{z_{2}}}
    &= \ip{ T_{z_{1}} g_{0} }{ T_{z_{2}} g_{0} } = \ip{ g_{0} }{ T_{-z_{1}} T_{z_{2}} g_{0} } \\
    &= \exp\parentheses{
      \tfrac{\i}{2\eps} \Omega(z_{1}, z_{2})
    } \ip{ g_{0} }{ T_{z_{2}-z_{1}} g_{0} } = \exp\parentheses{
      \tfrac{\i}{2\eps} \Omega(z_{1}, z_{2})
      -\tfrac{|z_{2} - z_{1}|^{2}}{4\eps}
      }. 
  \end{align*}
\end{proof}

By Lemma~\ref{lem:ip} we have
\begin{equation*}
  \ip{g_{z_{1}}}{g_{z}} \ip{g_{z}}{g_{z_{2}}} =
  \exp\parentheses{ -\tfrac{|z - z_{1}|^{2} + |z - z_{2}|^{2}}{4\eps} }
  \exp\parentheses{ \tfrac{\i}{2\eps} \Omega(z_{1} - z_{2}, z) }
\end{equation*}
and obtain for the Husimi function 
\begin{align*}
   \H_\psi(z)
  &= (2\pi\eps)^{-d} \left( 
    \exp\parentheses{
      -\tfrac{|z - z_{1}|^{2}}{2\eps}
    }
    + \exp\parentheses{
      -\tfrac{|z - z_{2}|^{2}}{2\eps}
    } \right. \\ \nonumber
    &  \qquad \left. +2 \exp\parentheses{
      -\tfrac{|z - z_{1}|^{2} + |z - z_{2}|^{2}}{4\eps}
    } \cos\parentheses{
    \tfrac{1}{2\eps} 
    \Omega(z_{1} - z_{2}, z)
  }\right).
\end{align*}
For the Hermite spectrograms we first use Lemma~\ref{lem:FBI} to obtain
\begin{align*}
  (\W_{\psi}*\W_{\varphi_{e_j}})(z)
  &= (2\pi\eps)^{-d} \; \abs{
    \sum_{k=1}^{2} \ip{g_{z}}{ \parentheses{A_{j} T_{z_{k}} - \tfrac{z^{\C}_{j}}{\sqrt{2\eps}} T_{z_{k}}}  g_{0} }
    }^{2}.
\end{align*}
We notice that, for $j = 1, \dots, d$ and $m = 1, 2$,
\begin{equation*}
  A_{j} T_{z_{m}} = \tfrac{z^{\C}_{m,j}}{\sqrt{2\eps}} T_{z_{m}} + T_{z_{m}} A_{j},
\end{equation*}
where $z^{\C}_{m,j} \defeq q_{m,j}+\i p_{m,j} \in \C$. This implies that
\begin{align*}
  \ip{g_{z}}{ \parentheses{A_{j} T_{z_{m}} - \tfrac{z^{\C}_{j}}{\sqrt{2\eps}} T_{z_{m}}} g_{0} } 
  &= \ip{g_{z}}{ \parentheses{ \tfrac{z^{\C}_{m,j} - z^{\C}_{j}}{\sqrt{2\eps}} T_{z_{m}} + T_{z_{m}} A_{j} } g_{0} } \\
  &= \tfrac{z^{\C}_{m,j} - z^{\C}_{j}}{\sqrt{2\eps}} \ip{g_{z}}{g_{z_{m}}}.
\end{align*}
Consequently, since $|z^{\C}_{m,j} - z^{\C}_{j}| = |z_{m,j} - z_{j}|$ for each $m$ and $j$, we have
\begin{align*}
(2\pi\eps)^{d}\,   (\W_{\psi}*\W_{\varphi_{e_j}})(z)
  &= \abs{
    \tfrac{z^{\C}_{1,j} - z^{\C}_{j}}{\sqrt{2\eps}} \ip{g_{z}}{g_{z_{1}}}
    + \tfrac{z^{\C}_{2,j} - z^{\C}_{j}}{\sqrt{2\eps}} \ip{g_{z}}{g_{z_{2}}}
    }^{2} \\
  &= \tfrac{1}{2\eps} \parentheses{
     |z_{1,j} - z_{j}|^{2} \abs{ \ip{g_{z}}{g_{z_{1}}} }^{2}
    + |z_{2,j} - z_{j}|^{2} \abs{ \ip{g_{z}}{g_{z_{2}}} }^{2}
    }\\
  &\quad+ \tfrac{1}{\eps}\Re\brackets{
     \overline{(z^{\C}_{1,j} - z^{\C}_{j})} (z^{\C}_{2,j} - z^{\C}_{j}) \ip{g_{z_{1}}}{g_{z}} \ip{g_{z}}{g_{z_{2}}}
  }.
\end{align*}
Now observe that
\begin{align*}
  \sum_{j=1}^{d} \overline{(z^{\C}_{1,j} - z^{\C}_{j})} (z^{\C}_{2,j} - z^{\C}_{j})
  &= \sum_{j=1}^{d} [ (q_{1,j} - q_{j}) - \i (p_{1,j} - p_{j}) ] [ (q_{2,j} - q_{j}) + \i (p_{2,j} - p_{j}) ] \\
  &= (z - z_{1}) \cdot (z - z_{2}) + \i\,\Omega(z - z_{1}, z - z_{2}).
\end{align*}
Therefore,
\begin{align*}
  &\sum_{j=1}^{d}(\W_{\psi}*\W_{\varphi_{e_j}})(z) \\ 
  &\quad= \tfrac{1}{(2\pi\eps)^{d}} \biggl\{
    \tfrac{|z - z_{1}|^{2}}{2\eps} \exp\parentheses{ -\tfrac{|z - z_{1}|^{2}}{2\eps} }
    + \tfrac{|z - z_{2}|^{2}}{2\eps} \exp\parentheses{ -\tfrac{|z - z_{2}|^{2}}{2\eps} } \\ 
  &\hspace{0.8in} + \tfrac{1}{\eps}
    \exp\parentheses{
    -\tfrac{|z - z_{1}|^{2} + |z - z_{2}|^{2}}{4\eps}
    }
    \Bigl[
    (z - z_{1}) \cdot (z - z_{2}) \cos\parentheses{ \tfrac{1}{2\eps} \Omega(z_{1} - z_{2}, z) } \\ 
  &\hspace{2.4in}
    - \Omega(z - z_{1}, z - z_{2}) \sin\parentheses{ \tfrac{1}{2\eps} \Omega(z_{1} - z_{2}, z) }
    \Bigr]
    \biggr\}.
\end{align*}
Combining the Husimi function and the Hermite spectrograms gives the density $\mu_\psi$ as plotted before in Figure~\ref{fig:superposition} in the Introduction.

\subsection{Hermite functions}
\label{ssec:Hermite_functions}
Let us consider the higher order Hermite functions
\[
\varphi_k \defeq \frac{1}{\sqrt{k!}} (A^\dagger)^k g_0,\qquad k=(k_1,\ldots,k_d)\in\N^d,
\]
which result from the $k$-fold application of the raising operator $A^\dagger$ to the Gaussian wave packet $g_0$ centered at the origin.
The FBI transform of a Hermite function is known as
\begin{equation}\label{eq:FBI}
(2\pi\eps)^{-d/2} \langle g_z,\varphi_k\rangle = 
\frac{1}{\sqrt{(2\pi\eps)^{d}k!}} \left(\tfrac{1}{\sqrt{2\eps}}(q-\i p)\right)^k \exp\!\left(-\tfrac{1}{4\eps}|z|^2\right)
\end{equation}
for $z\in\Rdd$; see~\cite[\S2]{F13} or \cite[Proposition 5]{LT14}. Hence, the Husimi function satisfies
\begin{equation}\label{eq:hu_he}
  \mathcal{H}_{\varphi_{k}}(z) = \frac{1}{(2\pi\eps)^{d} k!}  \left|\tfrac{1}{\sqrt{2\eps}}z\right|^{2k}    \,\exp\!\left(-\tfrac{1}{2\eps}|z|^2\right) 
  = \prod_{j=1}^d h_{k_j}(z_j),
\end{equation}
where 
\[
h_n(w)\defeq \H_{\varphi_n}(w),\qquad w\in\R^2,
\]
denotes the Husimi function of the univariate $n$-th Hermite function $\varphi_n$, $n\in\N$.
For the multivariate Hermite spectrograms we obtain the following: 

\begin{lemma}\label{lem:hermite} For all $k\in\N^d$ and $j=1,\ldots,d$ and $z\in\Rdd$, we have
\begin{align*}
& (\W_{\varphi_k}*\W_{\varphi_{e_j}})(z)\\
& \quad = \left(k_j h_{k_j-1}(z_j)-2k_jh_{k_j}(z_j)+(k_j+1)h_{k_j+1}(z_j)\right)\cdot \prod_{n\neq j} h_{k_n}(z_n).
\end{align*}
\end{lemma}

\begin{proof}
Since $A_j\varphi_k = \sqrt{k_j} \varphi_{k-e_j}$,  Proposition~\ref{lem:FBI} implies  
\begin{align*}
(\W_{\varphi_k}*\W_{\varphi_{e_j}})(z) &= 
(2\pi\eps)^{-d} \left| \langle g_z,A_j\varphi_k\rangle -\tfrac{1}{\sqrt{2\eps}}z_j^\C\langle g_z,\varphi_k\rangle\right|^2\\
&=
(2\pi\eps)^{-d} \left| \sqrt{k_j}\langle g_z,\varphi_{k-e_j}\rangle -\tfrac{1}{\sqrt{2\eps}}z_j^\C\langle g_z,\varphi_k\rangle\right|^2.
\end{align*}
By the formula \eqref{eq:FBI} for the FBI transform, we obtain
\begin{align*}
(\W_{\varphi_k}*\W_{\varphi_{e_j}})(z) &=
\frac{\exp\!\left(-\tfrac{1}{2\eps}|z|^2\right)}{(2\pi\eps)^{d}k!}  \left| k_j\left(\tfrac{1}{\sqrt{2\eps}}(q-\i p)\right)^{k-e_j} 
- \tfrac{1}{\sqrt{2\eps}}z_j^\C \left(\tfrac{1}{\sqrt{2\eps}}(q-\i p)\right)^{k}\right|^2\\
&=
\frac{\exp\!\left(-\tfrac{1}{2\eps}|z|^2\right)}{(2\pi\eps)^{d}k!} \left(k_j- \tfrac{1}{2\eps}|z_j|^2\right)^2
\left| \left(\tfrac{1}{\sqrt{2\eps}}z\right)^{k-e_j} \right|^2.
\end{align*}
Then it remains to observe that
\begin{align*}
&\frac{\exp\!\left(-\tfrac{1}{2\eps}z_{j}^2\right)}{(2\pi\eps)k_{j}!} \left(k_j- \tfrac{1}{2\eps}|z_j|^2\right)^2 \left|\tfrac{1}{\sqrt{2\eps}}z_j\right|^{2k_j-2}\\
 &\qquad = k_j h_{k_j-1}(z_j) - 2k_j h_{k_j}(z_j) + (k_j+1) h_{k_j+1}(z_j). 
\end{align*}
\end{proof}

The combination of the Husimi function \eqref{eq:hu_he} and the Hermite spectrograms of Lemma~\ref{lem:hermite} gives an explicit formula for the density $\mu_\psi$ when $\psi=\varphi_k$ for  $k\in\N^d$. 



\section{Numerical Experiments}\label{sec:numerics}

We present numerical experiments\footnote{All experiments have been performed with \textsc{Matlab} 8.3 on a $3.33$
  GHz Intel Xeon X5680 processor.} for computing expectation values for the solution of the time-dependent Schr\"odinger equation 
\[
\i\eps\partial_t \psi(t) = (-\tfrac{\eps^2}{2}\Delta + V)\psi(t),\qquad \psi(0)=\psi_0,
\]
with three different potentials $V:\R^d\to\R$. The various setups shall illustrate important aspects of our new algorithm, such as the second order
accuracy with respect to $\eps$, the good applicability in higher dimensions, and the
capability of describing fundamental quantum effects.

\subsection{Discretization}
For the algorithmic discretization of Corollary~\ref{thm:prop_husimi}
we proceed similarly  as in~\cite{LR10,GL14,KL13}. 
We consider various smooth functions $a:\Rdd\to\R$ and evaluate  the phase space integral on the right hand side of the semiclassical approximation
\begin{align*}
\left\langle \psi(t),\op(a) \psi(t)\right\rangle& = \int_\Rdd (a\circ\Phi_t)(z) \mu_{\psi_0}(z)\d z + O(\eps^2)
\end{align*}
for normalized initial data $\psi_0\in L^2(\R^d)$, $\|\psi_0\| = 1$, via the quadrature formula
\begin{align*}
\int_\Rdd (a\circ\Phi_t)(z) \mu_{\psi_0}(z)\d z & = (1+\tfrac{d}2)\int_\Rdd (a\circ\Phi_t)(z) \H_{\psi_0}(z)\d z \\
  &\quad - \tfrac{1}2  \sum_{j=1}^d \int_\Rdd (a\circ\Phi_t)(z)(\W_{\psi_0}*\W_{\varphi_{e_j}})(z)\d z \\
& \approx \frac{1+\tfrac{d}2}{N} \sum_{k=1}^N  (a\circ\Phi_t)(z_k) -\frac{d}{2N} \sum_{k=1}^N  (a\circ\Phi_t)(w_k),
\end{align*}
where one samples the quadrature points according to the probability measures
\[
z_1,\hdots,z_N \sim \H_{\psi_0}\,,
\qquad
w_1,\hdots ,w_N   \sim  \tfrac1d\sum_{j=1}^d \W_{\psi_0} * \W_{\varphi_{e_j}};
\]
see also~\S\ref{sec:gamma_sampling} for the sampling strategies used for the Hermite spectrograms. The rate of convergence for the above quadrature rule is proportional to $N^{-1/2}$ for Monte Carlo samplings. For low discrepancy (Quasi-Monte Carlo) sampling the convergence is faster, that is, of the order $\log(N)^{2d}/N$. However, the literature on non-uniform Quasi-Monte Carlo sampling is scarce, and it seems to be an open question whether the transformed Halton sequences employed in our numerical experiments form indeed a low discrepancy set or just come very close to being so in practice, see also \cite{AD14}.

For the discretization of the Hamiltonian flow $\Phi_t$, we apply the eighth-order symplectic splitting method 
from~\cite[Table 2.D]{Y90}, which is a suitable composition of the linear flows of
\[
\begin{cases} \dot q = p\\ \dot p =0\end{cases}\qquad \text{ and } \qquad  \begin{cases} \dot q = 0\\ \dot p =-\nabla V(q)\end{cases}.
\]
Since our algorithm evolves an ensemble of classical trajectories, the use of symplectic time integrators is crucial; see also \cite[Fig.~4.2]{LR10}.

\begin{figure}[ht!]
\includegraphics[width=\textwidth]{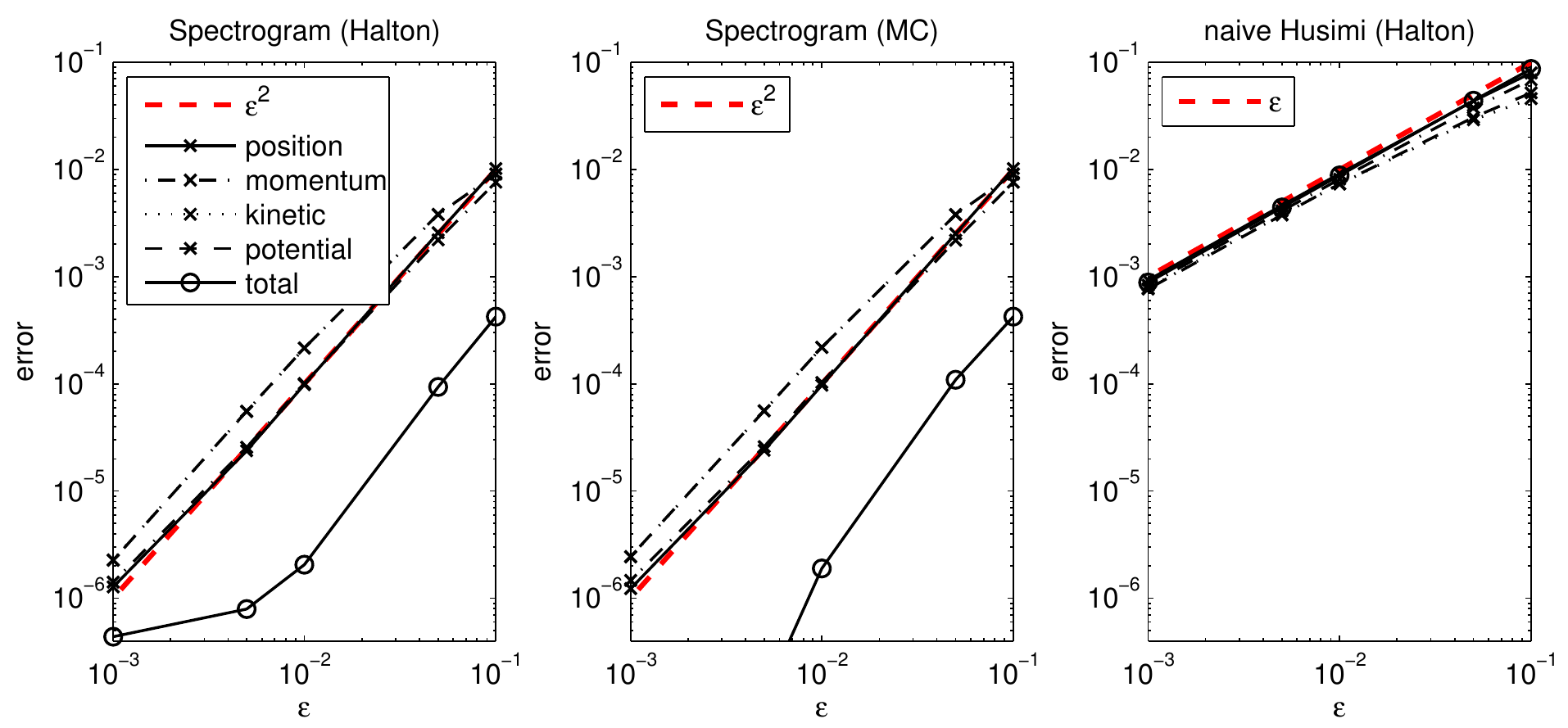}
\caption{\label{fig:gauss_error}
Average errors~\eqref{eq:aver_error} of the expectation values of various observables on the time interval $[0,20]$ for the new spectrogram method
 with initial Halton (left) and Monte Carlo (middle) sampling and results for the naive Husimi method with Halton sampling (right)
for the torsional potential and Gaussian initial data centered at $z = (1,0,0,0)$.}
\end{figure}

\subsection{Two-dimensional torsional potential}
Our first numerical experiments are conducted for the two-dimensional torsional potential
\[
V(q_1,q_2) = 2 - \cos(q_1) - \cos(q_2),\qquad q\in\R^2,
\]
and different values of the semiclassical parameter $\eps$. As the initial state we consider the Gaussian wave packet $\psi_0=g_z$ with phase space center $z = (1,0,0,0)$. 
This setup has already been considered in \cite{FGL09,LR10,KL13,GL14}. We investigate the dynamics of the following symbols $a:\R^4\to\R$,
\begin{enumerate}
\item Position: $a(q,p) = q_1$ and $a(q,p)=q_2$,
\item Momentum: $a(q,p)=p_1$ and $a(q,p)=p_2$,
\item Kinetic and potential energy: $a(q,p) = \tfrac12|p|^2$ and $a(q,p)=V(q)$,
\item Total energy: $a(q,p) = \tfrac12|p|^2 + V(q)$,
\end{enumerate}
and compare the outcome of the new algorithm with the naive, first-order Husimi approximation
\begin{equation}\label{eq:naive_Husimi}
\langle\psi(t),\op(a)\psi(t)\rangle = \int_{\Rdd} (a\circ\Phi_t)(z) \H_{\psi_0}(z) \d z + O(\eps).
\end{equation}

The left and middle panel of Figure~\ref{fig:gauss_error} confirm the second order accuracy of the new method for both 
Monte Carlo and Halton type samplings of the initial density~$\mu_{\psi_0}$.  The right panel illustrates that the naive Husimi method is indeed only of first order in $\eps$.
The time-averaged errors
\begin{equation}\label{eq:aver_error}
\frac{1}{20} \int_{0}^{20} \Big|\langle\psi_{\rm ref}(t),\op(a)\psi_{\rm ref}(t)\rangle -   \frac{1+\tfrac{d}2}{N} \sum_{k=1}^N  (a\circ\Phi_t)(z_k) + \frac{d}{2N} \sum_{k=1}^N  (a\circ\Phi_t)(w_k) \Big| \d t
\end{equation}
are taken with respect to highly accurate grid based reference solutions $ \psi_{\rm ref}(t)\approx \psi(t)$
of the Schr\"odinger equation; for details see Appendix~\ref{app:num_torsional}.

The total energy error in Figure~\ref{fig:gauss_error} is smaller 
than the errors for other expectation values. Firstly, this can be explained by
the fact that the total energy error is 
time independent, as explained in Remark~\ref{rem:tot_energy_error}, and
symplectic integrators are practically energy preserving
on the time scales considered here.
Secondly, the leading
$O(\eps^2)$ term in the asymptotic expansion of the error 
\begin{align*}
\langle &g_z,\op(h)g_z\rangle - \int_\Rdd h(w)\mu_{g_z}(w)dw = 
\int_\Rdd h(w)(\W_{g_z} - (1-  \tfrac{\eps}4 \Delta)\H_{g_z}) (w) dw \\
&=  \int_\Rdd h(w)\left(\W_{g_z} - (1-  \tfrac{\eps}4 \Delta)(1+ \tfrac{\eps}4 \Delta+ \tfrac{\eps^2}{32} \Delta^2)
\W_{g_z}\right) (w) dw  + O(\eps^3)\\
&=  \tfrac{\eps^2}{32} \int_\Rdd h(w)\Delta^2\W_{g_z}(w) dw + O(\eps^3)
\end{align*}
can be bounded by the small constant
\begin{align*}
\big|\tfrac{1}{32} \int_\Rdd \Delta^2h(w)\W_{g_z}(w) dw \big| \leq \tfrac{1}{16},
\end{align*} 
which follows from the special form of the torsional potential and the initial state.
The total energy errors for small values of $\eps$ are larger in the left panel of Figure~\ref{fig:gauss_error}
than in the middle panel, since the relatively small number of
Halton points gives rise to perceptible quadrature errors.

\begin{figure}[ht!]
\includegraphics[width=\textwidth]{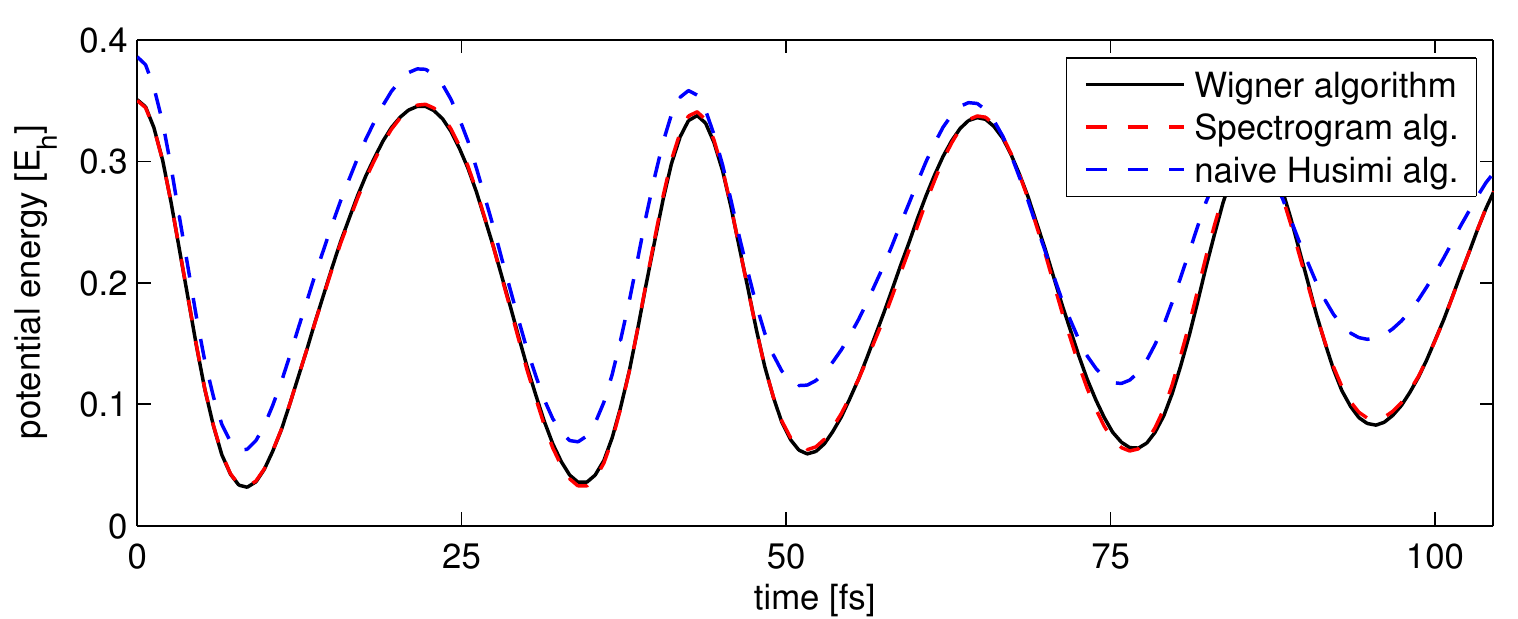}
\caption{\label{fig:henon_pot}
The approximate dynamics of the potential energies for the $32$-dimensional Henon--Heiles system
obtained by
the Wigner,  Spectrogram, and naive Husimi algorithms.}
\end{figure}

\subsection{Henon--Heiles dynamics for {\boldmath $d=32$}}
Henon--Heiles type systems have been used for benchmark simulations with the multiconfiguration time-dependent Hartree method (MCTDH); see \cite{NM02}. We follow the presentation in \cite[\S5.B]{KL14} by using
the potential 
\[
V_{32}(q) = \tfrac12 |q|^2 + 1.8436 \sum_{j=1}^{31} (q_j^2q_{j+1} - \tfrac13 q_{j+1}^3) + 0.4 \sum_{j=1}^{31} (q_j^2 + q_{j+1}^2)^2
\]
and the semiclassical parameter $\eps = 0.0029$, which is a model for the dynamics of a hydrogen atom
on a high-dimensional
potential energy surface that exhibits regions of chaotic motion. The quartic confinement 
guarantees that none of the classical trajectories
escapes to infinity. Moreover, as in~\cite{NM02,KL14}, the initial state is 
a Gaussian wave packet $\psi_0=g_z$
centered at $z=(q,p)$ with $p=0$ and $q_j = 0.1215$ for all $j=1,\ldots,32$. 

Since grid-based reference solutions are not available for this high-dimensional setting, we compare our
method with the results obtained by the second order approximation
\[
\langle \psi(t),\op(a)\psi(t)\rangle = \int_{\Rdd} (a\circ\Phi_{t})(z) \W_{\psi_0}(z) \d z + O(\eps^2).
\]
Note that, in this particular case, the initial Wigner function $\W_{\psi_0}(z)$ is the Gaussian~\eqref{eq:Wigner_for_Gaussian}, and hence the Wigner algorithm does not pose a difficulty in the initial sampling.
Figure~\ref{fig:henon_pot} shows the good agreement of the results from 
the Wigner and the spectrogram methods by means of the potential energy, and a considerable discrepancy with
respect to the naive Husimi approximation~\eqref{eq:naive_Husimi}.


\subsection{Escape from a cubic potential well}
Finally, we  explore whether the new method is capable of
describing the evolution of a quantum system that moves out of a potential well. For this purpose we consider the semiclassical Schr\"odinger operator $H=-\tfrac{\eps^2}{2}\Delta + V$
with the one-dimensional barrier potential
\[
V(q) = 2.328\cdot  q^2 +  q^3 + 0.025 q^4,\qquad q\in\R,
\]
and $\eps=0.4642$; see also Figure~\ref{fig:barrier}. This Hamiltonian can be derived from
the Schr\"odinger operator 
\[
 - \tfrac12\hbar^2 \Delta + \tfrac12 x^2 + 0.1x^3
\]
 with $\hbar=1$ from \cite{OL13,PP00} by applying the space rescaling $x\mapsto \sqrt[3]{0.1}x$
and adding the confinement term $0.025\cdot q^4$. The confinement prevents phase space
trajectories from finite time blow up and guarantees that $H$ is
essentially self-adjoint.
The global potential energy minimum $V(x_{\rm glob}) \approx -4765$ is attained at $x_{\rm glob}\approx -28.4$, and the confinement
 is very small in the region of interest close to the origin.

\begin{figure}[ht!]
\includegraphics[width=\textwidth]{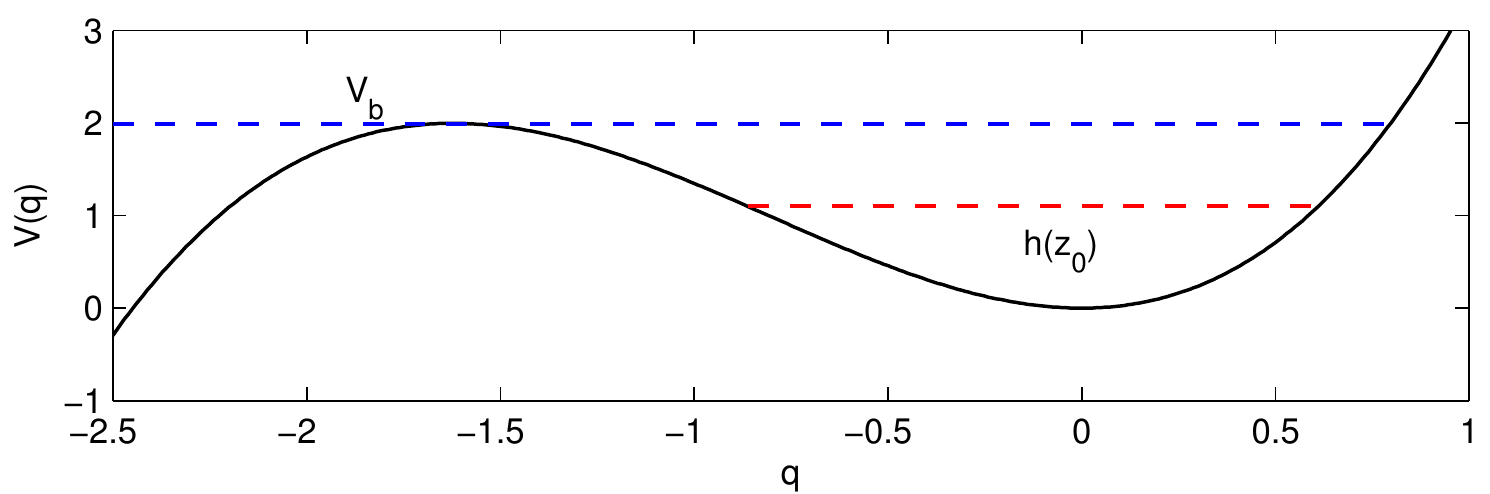}
\caption{\label{fig:barrier}
The cubic barrier potential $V$ with the barrier energy~$V_b$ and the energy $h(z_0)$ of the trapped classical particle.}
\end{figure}

As initial states we consider translated Hermite functions
\[
\psi_0 = T_{z_0}\varphi_k,\qquad k\in\{0,1,3,6\}
\] 
localized around
$z_0 = (0.4642,-1)$, which corresponds to the initial phase space center used in \cite{OL13}.
Since the associated classical energy $h(z_0)$ lies below the barrier energy $V_b\approx 2.03$,
 the classical particle is trapped in the well for all times. In contrast, the quantum energy 
of the initial state $g_{z_0}$ is approximately $2.09$, and the energies of the excited states
are even higher. Consequently, the expected phase space center of the quantum particle will escape
from the classical trapping region after short time. 

Figure~\ref{fig:tunneling} displays the trajectories of the expected phase space centers obtained by the purely classical, the full quantum, and the semiclassical spectrogram dynamics for the four different initial states. 
The results from the spectrogram algorithm show decent qualitative agreement  with the 
behavior of the quantum solution even though the semiclassical parameter $\eps=0.4642$ is rather large. We also note that the results become more accurate for initial states of higher energy.

\begin{figure}[ht!]
\includegraphics[width=\textwidth]{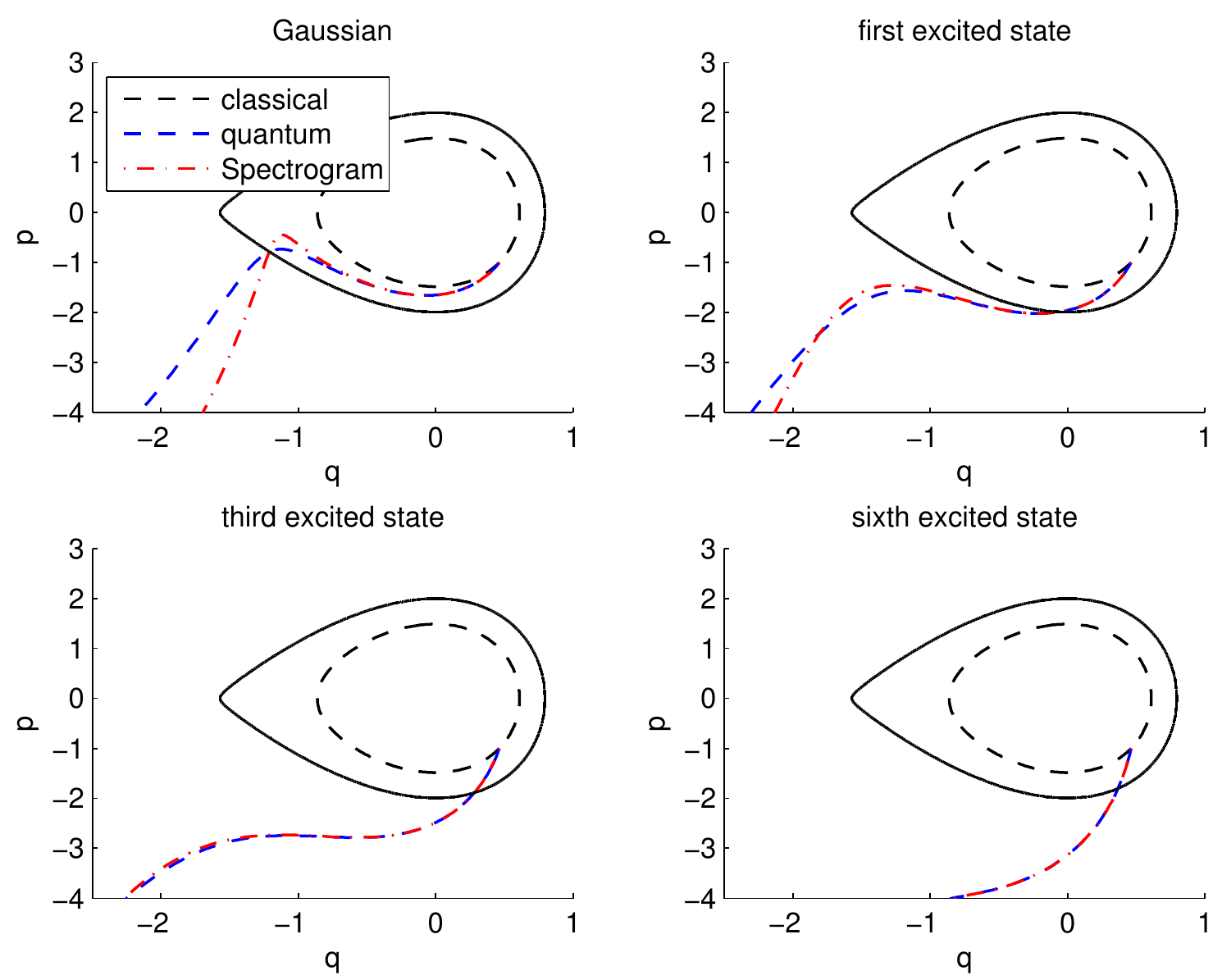}
\caption{\label{fig:tunneling}
Trajectories of the expected phase space centers obtained from quantum references and the new spectrogram algorithm
for different initial states $\psi_0=T_{z_0}\varphi_k$.
The dashed  black line shows the periodic classical orbit associated with $z_0$, and the solid black line illustrates
the border of  the trapping region.}
\end{figure}

The spectrogram algorithm is also capable of describing the evolution of the probability that the quantum particle escapes the potential well 
and is found in the region $(-\infty,x_{\rm max}]$, where $x_{\rm max} \approx -1.62$ is
the local maximum of the barrier potential.
To illustrate this property, we introduce the approximate escape probability
\[
P(t) = \left\langle \psi_t,\op(r) \psi_t\right\rangle \approx \|\psi_t \chi_{(-\infty, x_{\rm max}]} \|^2_{L^2}
\]
with the smooth symbol $r(q,p) = \exp(-0.01/(q-x_{\rm max})^2)\chi_{(-\infty, x_{\rm max}]}(q)$, where $\chi_A$ denotes
the characteristic function of the set $A$. $P(t)$ can easily be approximated by
the spectrogram algorithm, and accurate numerical references are available;
see Appendix~\ref{sec:ref_well}.
\begin{figure}[ht!]
\includegraphics[width=\textwidth]{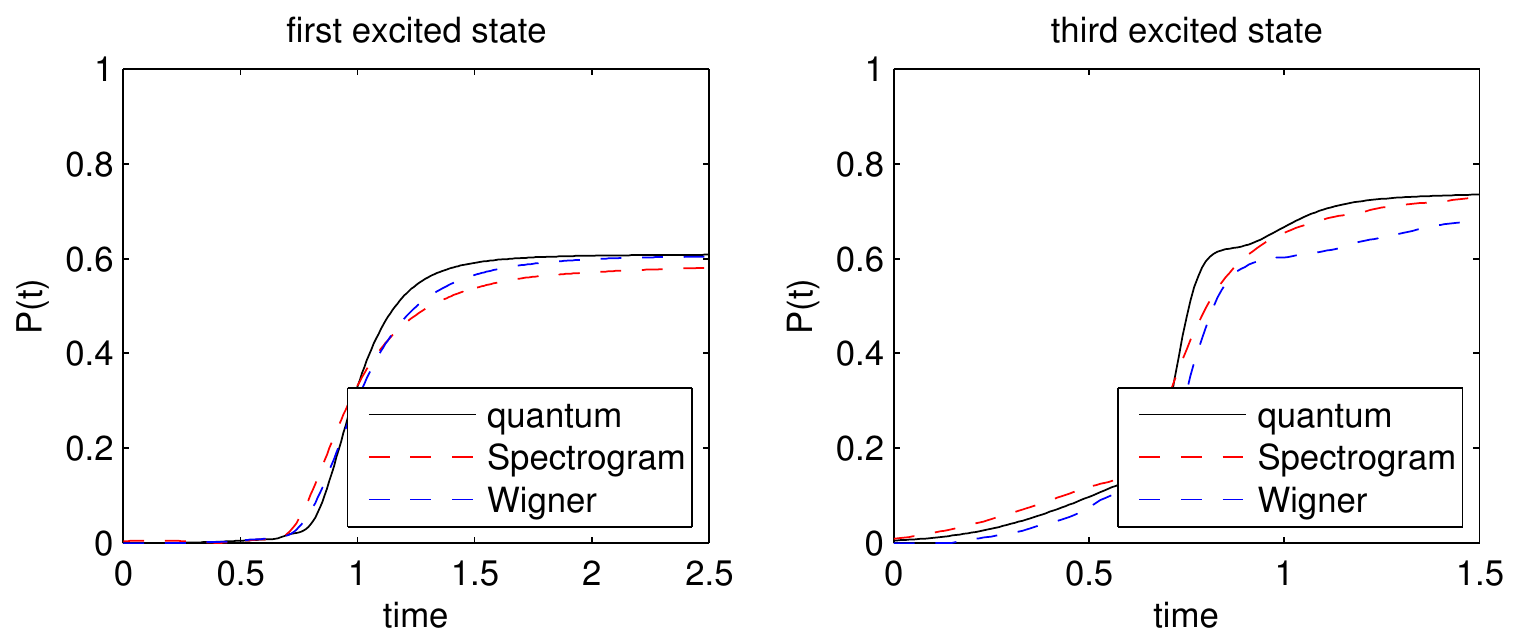}
\caption{\label{fig:tunnel_prob}
Approximate escape probabilities $P(t)$ computed from a highly accurate numerical quantum reference, and results of the
spectrogram and Wigner algorithms for the initial states $T_{z_0}\varphi_1$, and $T_{z_0}\varphi_3$.}
\end{figure}
Figure~\ref{fig:tunnel_prob} shows by means of two different initial states 
that the spectrogram algorithm 
yields a good qualitative picture of the evolution of escape probabilities.

\appendix
\section{Sampling by the Gamma distribution}\label{sec:gamma_sampling}
\subsection{Using the Gamma distribution}
We consider the Hermite functions translated by the Heisenberg--Weyl operator, 
$$
\psi=T_z\varphi_k, \qquad k\in\N^d, \qquad z\in\R^{2d}.
$$ 
Then, by the covariance property \eqref{eq:spectrogram-translated}, the Hermite spectrogram takes the form
\begin{equation*}
  (\W_{T_{z}\varphi_{k}}*\W_{\varphi_{e_j}})(w) = (\W_{\varphi_{k}}*\W_{\varphi_{e_j}})(w - z), \qquad w\in\Rdd,
\end{equation*}
and by Lemma~\ref{lem:hermite},  we only have to consider the two-dimensional probability densities
$$
w_j\mapsto h_n(w_j-z_j) = \frac{1}{2\pi\eps\cdot n!}\left|\tfrac{1}{\sqrt{2\eps}}(w_j-z_j)\right|^{2n} \exp\!\left(-\tfrac{1}{2\eps}|w_j-z_j|^2\right)
$$
with $n\in\{k_1,\ldots,k_d\}$.
Specifically, we first translate by $-z_j$ and then use the uniform distribution on $[0,2\pi]$ for the angular part combined with sampling from
$$
r\mapsto \frac{r^{2n+1}}{2^n \eps^{n+1} n!}  \exp\!\left(-\tfrac{1}{2\eps}r^2\right)
$$ 
for the radial part. Since
$$
\int_a^b   \frac{r^{2n+1}}{2^n \eps^{n+1} n!} \exp\!\left(-\tfrac{1}{2\eps}r^2\right) \d r =  \int_{a^2}^{b^2} \frac{\tau^n}{(2\eps)^{n+1}n!} \exp\!\left(-\tfrac{1}{2\eps}\tau\right) \d \tau
$$
for all $a,b>0$, we use the Gamma distribution with parameters $n+1$ and $2\eps$ for the radial sampling. 

In the particular case $k = 0$, $\varphi_{0} = g_{0}$ and so $T_{z}\varphi_{0} = g_{z}$ is a Gaussian wave packet and
$$
\mu_{g_z}(w) =  (2\pi\eps)^{-d}\left(1+\tfrac{d}{2}-\tfrac{1}{4\eps}|w-z|^2\right) \exp\!\left(-\tfrac{1}{2\eps}|w-z|^2\right),\qquad w\in\Rdd;
$$
see \S\ref{sec:gaussian}. One can directly sample this $2d$-dimensional probability density without factorizing its summands.
After translation by $-z$, one decomposes into a radial part and an independent angular variable which is uniformly 
distributed on $\mathbb{S}^{2d-1}$. The radial density is given by
\[
r \mapsto \frac{r^d}{(2\eps)^{d+1}d!} \exp\!\left(-r/2\eps\right),
\]
which corresponds to a Gamma distribution with parameters $d+1$ and $2\eps$.

\subsection{Monte Carlo sampling}\label{sec:MC}
Pseudorandom numbers uniformly distributed on a multi-dimensional unit sphere are obtained by sampling from multivariate normal 
distributions and subsequent normalization, while Gamma distributed pseudorandom samples only require the uniform distribution on 
the unit cube together with the (numerical) inverse of the cumulative Gamma distribution function. Hence a Monte-Carlo sampling of 
$\mu_{T_z \varphi_k}$ is straightforward.

\subsection{Quasi-Monte Carlo sampling}
For the generation of Quasi-Monte Carlo points we have heuristically mimicked the 
Monte-Carlo procedure of \S\ref{sec:MC} by replacing the pseudorandom samples from the uniform distribution on the unit cube by a Halton 
sequence. Whether the resulting points are of low discrepancy with respect to the distribution $\mu_{T_z\varphi_k}$ seems to 
be an open question not answered by the current literature, see e.g. \cite{AD14}. 

\newpage
\section{Numerical data}\label{app:num}
\subsection{Two-dimensional torsional system}
\label{app:num_torsional}
Table~\ref{tab:gauss_comp} contains the number of initial sampling points and the computational time for the new spectrogram algorithm. 
In the case of Monte Carlo integration, the results in Figure~\ref{fig:gauss_error} are averaged over ten independent runs.
For the time propagation we apply the above mentioned eighth-order symplectic integrator with time stepping~$10^{-1}$. 
The parameters of the grid-based reference solver are collected in Table~\ref{tab:reference_torsion}.

\begin{table}[ht!]
\centering
\begin{tabular}{c|| r | r|| r|r}
$\eps$ & MC points & comp. time & Halton points & comp. time \\ \hline
$10^{-1}$&$5\cdot 10^4 $&23s & $5\cdot 10^4$ & 16s \\
$5\cdot10^{-2}$&$ 3\cdot 10^5$&1m59s & $ 10^5$& 33s\\
$10^{-2}$& $6 \cdot 10^5$& 7m16s & $2\cdot 10^5$& 1m59s\\
$5\cdot 10^{-3}$& $1.5\cdot 10^6$& 14m15s & $8\cdot 10^5$& 6m50s\\
$10^{-3}$&$10 \cdot 10^6 $& 68m30s& $2\cdot 10^6$& 18m31s
\end{tabular}
\medskip
\caption{\label{tab:gauss_comp}  Computational data for the execution of the new spectrogram algorithm
for the  two dimensional torsional potential and Gaussian initial data on
the time interval $[0,20]$. The computation times are for one run only. They scale linearly with respect to the number of initial sampling points.}
\end{table}

\begin{table}[ht!]
\centering
\begin{tabular}{c|ccc}
$\eps$    & $\#$timesteps &comp. domain               & space grid\\\hline 
$10^{-1}$ & $5\cdot10^3$   & $[-3,3]\times[-3,3]$ & $1536\times 1536$\\
$5\cdot 10^{-2}$ & $5\cdot10^3$   & $[-3,3]\times[-3,3]$ & $1536\times 1536$\\
$10^{-2}$ & $7.5\cdot10^3$   & $[-2,2]\times[-2,2]$ & $2048\times 2048$\\
$5 \cdot 10^{-3}$ & $10^4$         & $[-2,2]\times[-2,2]$ & $2048\times 2048$\\
$10^{-3}$ & $10^4$         & $[-2,2]\times[-2,2]$ & $2048\times 2048$
\end{tabular}
\medskip
\caption{\label{tab:reference_torsion} Parameters of the grid-based reference solutions for the two-dimensional torsional potential.
The discretization has been done by Fourier collocation in $\R^2$ and Strang splitting in time.
}
\end{table}

\subsection{Henon--Heiles system for {\boldmath $d=32$}}
For the initial sampling we used $2^{17}$
Halton type points for all three initial densities, that is, the Wigner and the Husimi functions and the new density $\mu_{\psi_0}$. The time stepping of the eighth order time integrator is~$2\cdot 10^{-2}$.

\subsection{One-dimensional cubic well}\label{sec:ref_well}
For the spectrogram algorithm
we employed $2^{14}$ Halton points and the eighth-order integrator with time stepping $10^{-2}$,
which results in a computational time of $2$ seconds. The quantum references are generated by means of a 
Strang splitting with $2^{15}$ Fourier modes on the interval $[-40,4]$.

\bibliographystyle{siam}


\end{document}